\documentclass[11pt]{article}

\usepackage{sn-preamble} 
\usepackage{tikz}
\usepackage{verbatim}
\usepackage{cancel}
\usepackage{enumitem}
\usepackage{bm}
\usepackage{caption}
\usepackage{microtype}

\newcommand{\spnote}[1]{}
\newcommand{\snnote}[1]{}
\newcommand{\shyamal}[1]{}
\newcommand{\shivam}[1]{}

\newcommand{\ignore}[1]{}

\newcommand{\coords}{\textsc{Construct-Coordinate-Oracles}}

\newcommand{\yes}{{\mathrm{yes}}}
\newcommand{\no}{{\mathrm{no}}}
\newcommand{\Dyes}{\calD_{\yes}}
\newcommand{\Dno}{\calD_{\no}}
\newcommand{\bfyes}{\boldf_\yes}
\newcommand{\bfno}{\boldf_\no}

\newcommand{\fno}{f_\no}

\newcommand{\warmbt}{\textsc{Warmup-Ball-Tester}}
\newcommand{\finebt}{\textsc{$k$-Junta-Distance}}

\newcommand{\TpS}{\textsc{Hold-Out-Noise-Evaluations}}

\newcommand{\Bad}{\mathsf{Bad}}

\title{Optimal Non-Adaptive Tolerant Junta Testing\\ via Local Estimators \vspace{0.5em}}
\author{
Shivam Nadimpalli\\ Columbia University\\ \normalsize \texttt{sn2855@columbia.edu}
\and 
Shyamal Patel\\ Columbia University\\ \normalsize \texttt{shyamalpatelb@gmail.com}
\vspace{0.5em}
}
\date{\small\today}

\begin{document}

\pagenumbering{gobble}
\maketitle  

\begin{abstract}

We give a non-adaptive algorithm that makes $2^{\wt{O}(\sqrt{k\log(1/\eps_2 - \eps_1)})}$ queries to a Boolean function $f\isafunc$ and distinguishes between $f$ being $\eps_1$-close to some $k$-junta versus $\eps_2$-far from every $k$-junta. At the heart of our algorithm is a {local} mean estimation procedure for Boolean functions that may be of independent interest. We complement our upper bound with a matching lower bound, improving a recent lower bound obtained by Chen et al. We thus obtain the first tight bounds for a natural property of Boolean functions in the tolerant testing model.

\end{abstract}

\newpage

\pagenumbering{arabic}

\section{Introduction}
\label{sec:intro}

A Boolean function $f\isafunc$ is a \emph{$k$-junta} if its output depends on only $k$ out of its $n$ input variables. Juntas are a central object of study in computational complexity theory~\cite{Jukna01book,Juknacomplexitybook,odonnell-book} and related areas such as learning theory~\cite{Blum:94relevant, MOS:03}, where they elegantly model the problem of learning in the presence of irrelevant features. 

\paragraph{Junta Testing.}

Consider the problem of testing juntas: Given query access to a Boolean function $f\isafunc$, distinguish with probability $2/3$ whether (i) $f$ is a $k$-junta; or (ii) $f$ is $\eps$-far from \emph{every} $k$-junta, where we say that $f$ is \emph{$\eps$-far} from $g$ if 
\[\dist(f,g) := \Prx_{\bx \sim \bits^n} [g(\bx) \not = f(\bx)] \geq \eps.\]
We will say that $f$ is \emph{$\eps$-close} to $g$ if $\dist(f,g) \leq \eps$.
Testing $1$-juntas (i.e. ``dictators'') has its origins in the study of PCPs~\cite{BGS98,Hastad:01}, and the general problem of testing $k$-juntas was developed in~\cite{parnas2001proclaiming,FKR+:02}. 
After two decades of intensive research, the complexity of testing $k$-juntas is well-understood: 
\begin{itemize}
	\item The state-of-the-art {adaptive} algorithms use $\wt{O}(k/\eps)$ queries~\cite{Blaisstoc09,Bshouty19}, and matching $\wt{\Omega}(k)$ query bounds are known for adaptive algorithms when $\eps$ is a sufficiently small constant~\cite{ChocklerGutfreund:04, sauglam2018near}.
	\item For {non-adaptive} algorithms, Blais~\cite{Blaisstoc09}  gave a $\wt{O}(k^{3/2}/\eps)$-query algorithm, and a celebrated result of \cite{chen2018settling} proved a matching lower bound of $\wt{\Omega}(k^{3/2}/\eps)$ queries against non-adaptive algorithms (see also~\cite{servedio2015adaptivity}). 
\end{itemize}

We remind the reader that an algorithm is \emph{adaptive} if during its execution its choice of queries to $f$ are allowed to depend on the answers to the queries made thus far; we say that it is non-adaptive otherwise. In other words, the queries made by a non-adaptive algorithm are independent of the function $f$. Non-adaptive algorithms are frequently preferred over their adaptive counterparts, in large part due to their simpler as well as highly parallelizable nature.

\paragraph{{Tolerant} Junta Testing.} 
Note that the standard property testing model is extremely brittle: It requires the algorithm accept if and only if the function satisfies the property. This is not desirable in many applications, where the presence of noise in the queries to $f$ allow a tester in the standard model to simply reject the function. As a concrete example, it is often the case when learning a function that a few variables may explain most of the behavior of our function, but not all of it. In this case, we would again morally like to view our function as a junta. 

Motivated by this, Parnas, Ron and Rubinfeld~\cite{parnas2006tolerant} introduced the model of \emph{tolerant} property testing. The tolerant junta testing problem---which is the focus of this paper---is the following: 
Given query access to a function $f\isafunc$ and constants $0\leq \eps_1 < \eps_2 \leq 1/2$, distinguish with probability $2/3$ whether (i) $f$ is $\eps_1$-close to some $k$-junta; or (ii) $f$ is $\eps_2$-far from every $k$ junta. We will say that an algorithm $(\eps_1, \eps_2)$-tolerantly tests $k$-juntas if it has this performance guarantee. Note that the case when $\eps_1 = 0$ recovers the standard property testing model.

\begin{remark}
\label{rem:tol-test-dist-equiv}
It is not too difficult to see that tolerant junta testing is equivalent to the problem of \emph{distance estimation to a junta}. More formally, writing $\calJ_k$ for the class of $k$-juntas on $n$ variables and for $f\isafunc$, defining 
\[\dist(f,\calJ_k) := \min_{g \in \calJ_k} \dist(f,g),\]
one can see that estimating $\dist(f,\calJ_k)$ up to additive error $\eps := (\eps_2-\eps_1)/2$ immediately gives a $(\eps_1, \eps_2)$-tolerant tester~\cite{parnas2006tolerant}.
In light of this, we will frequently switch between tolerant testing juntas and junta distance estimation.
\end{remark}

The landscape of tolerant junta testing is starkly different from that of junta testing in the standard model:
\begin{itemize}
	\item The state-of-the-art upper bound, due to Iyer, Tal, and Whitmeyer~\cite{ITW21} gives an adaptive $2^{\wt{O}(\sqrt{k/\eps})}$-query algorithm for estimating distance to a $k$-junta. On the lower bounds front, the sole improvement to lower bounds from the standard model is due to the recent work of~\cite{ChenPatel23} that shows a mere $k^{\Omega(\log(1/\eps))}$ queries are necessary to approximate the distance to a $k$-junta. 
	\item On the non-adaptive front, the best known upper bound is due to De, Mossel, and Neeman~\cite{DMN19} who gave a $2^{k} \cdot \poly(k, \eps^{-1})$-query algorithm to estimate the distance to junta to additive $\eps$-error. (See also prior work by~\cite{parnas2006tolerant,chakraborty2012junto,BCELR16}.) Turning to lower bounds, a recent line of work starting with~\cite{LW18} yielded a $2^{\Omega(\sqrt{k})}$-query lower bound for some constant~$\eps=\Theta(1)$~\cite{chen2023mildly}. The result of~\cite{chen2023mildly} closely builds upon a previous lower bound due to~\cite{PRW22}.
\end{itemize}

Our main result closes the gap (up to logarithmic factors) between the upper and lower bounds on the query complexity of non-adaptive tolerant junta testing:

\begin{theorem}
\label{thm:fine-tester}
There exists a non-adaptive $\eps$-distance estimator for the set of $k$-juntas that makes at most $\poly(k, \eps^{-1}) \cdot 2^{\wt{O}(\sqrt{k \log(1/\eps)})}$ queries, where the $\wt{O}$ notation hides $\log(k)$ and $\log\log(1/\eps)$ factors.
\end{theorem}

In particular, our non-adaptive algorithm matches the ``highly adaptive'' state-of-the-art upper bound~\cite{ITW21} in terms of dependence on $k$, and even improves upon its $\eps$-dependence. 
Our main technical insight is to import ideas and techniques from approximation theory---in particular, tools used to obtain ``approximate inclusion-exclusion'' bounds~\cite{LinialNisan:90,kahn1996inclusion}---to junta testing. 
Using this framework, we construct estimators for the absolute value of the mean of a function $f\isafunc$ (i.e. for $|\E[f]|$) that can be \emph{locally} computed, i.e. computed using only the values of the function restricted to a random Hamming ball of radius $O(\sqrt{n})$. 
We believe that our local estimator is of independent interest.
Furthermore, to our knowledge, this is the first application of these tools to property testing of Boolean functions. 

We additionally complement our upper bound with a matching lower bound, improving upon the construction of \cite{chen2023mildly} by incorporating general $\eps$ dependence.

\begin{theorem}
\label{thm:lower-bound}
	Let $2^{-O(k)} \leq \eps$, then any $\eps$-distance estimator for $k$-junta must make at least $2^{\wt{\Omega}(\sqrt{k \log(1/\eps)})}$ queries, where the $\wt{\Omega}$ hides $\log(k)$ and $\log \log \eps^{-1}$ factors.
\end{theorem}

Note that the restriction that $\eps \geq 2^{-O(k)}$ is necessary: When $\eps < 2^k$, the tester of \cite{DMN19} only makes $2^k \poly(k, \eps^{-1})$, which is polynomial in $\eps^{-1}$. Together, \Cref{thm:fine-tester,thm:lower-bound} settle the query complexity of non-adaptively, tolerantly testing $k$-juntas. To the best of our knowledge, this is the first natural tolerant testing question for which tight bounds are known. 

\begin{table}[]
\centering
\renewcommand{\arraystretch}{1.5}
\setlength{\tabcolsep}{20pt}
\begin{tabular}{@{}llll@{}}
\toprule
               & $(\eps_1, \eps_2)$             & Query Complexity                         & Type         \\ \midrule
\cite{BCELR19} & $\pbra{\frac{\eps}{16}, \eps}$ & $\exp\pbra{O(k)}/\eps$                           & Non-Adaptive \\
\cite{DMN19}   & $(c-\eps, c)$                  & $\exp\pbra{O(k)}\cdot\poly\pbra{\frac{1}{\eps}}$ & Non-Adaptive \\
\cite{ITW21}   & $(c-\eps, c)$                  & $\exp\pbra{\wt{O}(\sqrt{k/\epsilon})}$            & Adaptive     \\
\Cref{thm:fine-tester}  & $(c-\eps, c)$         & $\exp\pbra{\wt{O}(\sqrt{k\log({1}/{\eps})}}$                                         & Non-Adaptive \\
\bottomrule
\end{tabular}
\caption{A summary of upper bounds for tolerant $k$-junta testing. Here $c,\eps \in [0, 0.5]$.}
\label{tab:ub-summary}
\end{table}

\subsection{Technical Overview}
\label{subsec:overview}

We now turn to a technical overview of our results. 

\paragraph{Tolerant Junta Testing via Local Estimators.}

Our results are motivated by a simple observation: 
\begin{quotation}
\noindent Previous lower bounds for tolerant junta testing~\cite{PRW22,chen2023mildly} typically involve constructing two distributions over Boolean functions that have different expected means, but look identical on Hamming balls of small radius. 
\end{quotation}
Indeed, a connection between the mean of a function and distance to junta is direct as for a set $S \subseteq [n]$ with $|S| = k$, we have that 
\begin{equation} \label{eq:dist-mean-abs-val}
	\dist(f, \calJ_S) = \frac{1}{2} - \frac{1}{2} \Ex_{\bx \in \bits^S} \sbra{\left| \Ex _{\by \in \bits^{[n] \setminus S}} [f(\bx \sqcup \by)]\right|}.
\end{equation}
(Here, $\dist(f,\calJ_S)$ denotes the minimum distance of $f$ to a junta on the variables in $S$.) 

These lower bound constructions, together with~\Cref{eq:dist-mean-abs-val}, motivate the following question of \emph{local mean estimation}: Given function $f\isafunc$ and access to values of $f$ restricted to a Hamming ball of radius $r$ centered at a random point $\bx\in\bn$, can you obtain a good estimate for the mean of the function (i.e. with probability $1-\delta$, the output is within $\pm \eps$ of the true mean)?

Naturally, we would like to do this with $r$ being as small as possible. Let $B(x,r)$ denote a Hamming ball of radius $r$ centered at $x$. As a starting point, any estimator that sees a ball $B(x,r)$ where $f$ is constant must output a number in $[f(x)\pm\epsilon]$, as otherwise the estimator would fail on the constant function. Taking $f$ to be the $n$-bit Majority function and $r = 0.001\sqrt{n}$, however, we see that most balls will be constant despite the Majority function having mean $0$. It follows that we must take $r = \Omega(\sqrt{n})$. Our key technical contribution establishes that for constant $\eps$ and $\delta$, Hamming balls of radius $O(\sqrt{n})$ are also sufficient to obtain good estimates of the mean. We prove this by combining constructions of ``flat polynomials'' (motivated by the problem of approximate inclusion-exclusion\cite{LinialNisan:90, kahn1996inclusion}) together with Fourier analysis of Boolean functions~\cite{odonnell-book}.

Using this local estimator, we can then get a junta tester when $n = 2k$ (cf. \Cref{sec:warmup}). Namely, we sample a random $y \in \bits^{2k}$ and estimate the means of the functions $f(y|_S \sqcup x): \bits^k \rightarrow \bits$ for each set $S \subseteq [2k]$ of size $k$ by querying $f$ on each point in $B(y,r)$. These in turn allow us to compute the distance of $f$ to $\calJ_S$ for every subset $S$ of size $k$. 

For larger $n = \poly(k)$, we need a better local mean estimator. Indeed, by the above lower bound determining the mean of a function on say $k^{100}$ variables requires a ball of radius $k^{50}$. Querying such a ball, however, would require far too many queries. Fortunately, we show that balls of smaller radius suffice when the function at hand has exponentially decaying Fourier tails. (Intuitively, one should consider the case of applying noise to a function, which should smooth it and make local balls more indicative of the mean of the function.) With this in hand, we noise the function and test in the same way as in the $O(k)$ setting. As before, the query complexity is dominated by the number of queries to determine the value of $f$ on the ball, which is at most $\poly(k)^{\sqrt{k}} = 2^{\wt{O}(\sqrt{k})}$.

Unfortunately, for general $n$ querying the ball of radius $\sqrt{k}$ would require us to make $2^{\sqrt{k} \log(n)}$ queries, which has an undesirable dependence on $n$. To circumvent this, we show that our estimator can equivalently written as linear combinations of higher-order derivatives of the noise operator. We can then evaluate these using high-precision numerical differentiation and get a good estimate of mean.

\paragraph{A Matching Lower Bound.} The lower bound construction is a simple modification of the lower bound of \cite{chen2023mildly}, which in turn builds on the construction of \cite{PRW22}. Essentially to get a lower bound of $\exp(\Omega(r))$ one must construct a pair of distributions $\calD_1$ and $\calD_2$ over functions $g: \bits^{k} \rightarrow \bits$. Crucially, we need that (i) the functions must have very different means i.e. 
\[\left| \Ex_{\mathbf{g} \sim \calD_1} [\mathbf{g}] \right| - \left| \Ex_{\mathbf{g} \sim \calD_2} [\mathbf{g}] \right| \geq \eps,\] 
and (ii) the functions should be identically distributed over balls meaning for all $y \in \bits^{|B(x,r)|}$ we have that for all $x \in \{0,1\}^k$ we have 
\[\Prx_{\mathbf{g} \sim \calD_1} \sbra{\mathbf{g}(z) = y_z~ \text{for all}~ z \in B(x,r)} = \Prx_{\mathbf{g} \sim \calD_2} \sbra{\mathbf{g}(z) = y_z~ \text{for all}~ z \in B(x,r)}.\]

In their paper, \cite{PRW22} gives a simple construction that are identically distributed on balls of radii $O(\sqrt{n})$. By breaking symmetry, we show that a simple modification can yield a gap of $O(\sqrt{n \log(1/\eps)}$. Using this, we can follow the analysis of \cite{chen2023mildly} to prove the result.

\subsection{Discussion}
\label{subsec:discussion}

Our work raises a number of intriguing directions for future work; we briefly describe some of them below:

\begin{itemize}
	\item Does adaptivity help for tolerant junta testing? While this paper resolves the tolerant junta testing question for non-adaptive algorithms, there are still large gaps for adaptive algorithms. As such, it is unclear how many queries are needed to tolerantly test $k$-juntas with adaptivity.
	\item Can we obtain improved runtime for tolerant junta testing? While our algorithms are query optimal, they all require time $\exp(k) \cdot 2^{\wt{O}(k \sqrt{\log(1/\eps)})}$. Is it possible have them run in time polynomial in the number of queries?
	\item Finally, a broad direction is that of potential applications of local estimators to other problems in algorithms and complexity theory. In particular, our approach hints at possible connections to pseudorandomness: Our local estimators imply that for any balanced $\textsf{AC}^0$ circuit $f$ and almost all $x \in \bits^n$ there exists a $y \in \bits$ with $f(x) \not = f(y)$ and $\dist(x,y) \leq \polylog(n)$.
\end{itemize}

\section{Preliminaries}
\label{sec:prelims}

We use boldfaced letters such as $\bb, \bx, \boldf, \bA$, etc. to denote random variables (which may be real-valued, vector-valued, function-valued, or set-valued; the intended type will be clear from the context).
We write $\bx\sim\calD$ to indicate that the random variable $\bx$ is distributed according to probability distribution $\calD.$ 

Given a set $J\sse[n]$, we will write $\overline{J} := [n]\setminus J$ to denote its complement. We will denote write $\mathbf{1}\cbra{\cdot}$ for the indicator function of the event $\cbra{\cdot}$.

\subsection{Boolean Functions}
\label{subsec:boolean-funcs}

Given Boolean functions $f,g\isafunc$ and a class of Boolean functions $\calC$, we define 
\[\dist(f,g) := \Prx_{\bx\sim\bn}\sbra{f(\bx)\neq g(\bx)}
\qquad\text{and}\qquad 
\dist(f,\calC) := \min_{g \in \calC}\dist(f,g).\]
In particular, if $\dist(f, \calC) \leq \eps$, then we say that $f$ is ``$\eps$-close'' to $\calC$; otherwise, we say that it is ``$\eps$-far'' from $\calC$.
Given a set $J\sse[n]$, and $z\in\bn$, we define the restricted function
\[f|_{J\to z} : \bits^{\overline{J}} \to \bits\] as
$f|_{J\to z}(x) := f(z_J, x_{\overline{J}}).$
Here we identify $\bn$ with $\bits^J\times\bits^{\overline{J}}$ in the natural fashion.

\begin{definition} \label{def:junta}
	A function $f\isafunc$ is a \emph{$k$-junta} if it only depends on $k$ out of its $n$ variables, i.e. if there exists a function $g\isafunc$ and indices $i_1, \ldots, i_k \in [n]$ such that  
\[f(x) = g(x_{i_1},\ldots, x_{i_k}).\]
We will write $\calJ_k$ for the class of $k$-juntas over $\bn$ where $n$ will be clear from context; similarly, given a set $J \sse[n]$, we will write $\calJ_J$ to denote the class of juntas on the variables in the set $J$.
\end{definition} 

\begin{notation} 
Given a point $x\in\bn$ and $r\in\N$, we write $B(x,r)$ for the Hamming ball of radius $r$ centered at $x$, i.e. 
\[B(x,r) := \cbra{y\in\bn : |\cbra{i : x_i\neq y_i}| \leq r}.\]
Given $x\in\bn$ and $T\sse[n]$, we write $x^{\oplus T}$ for the point obtained by flipping the bits of $x$ indexed by $T$.
\end{notation}


\subsection{Fourier Analysis over $\bn$}
\label{subsec:fourier-101}

Our notation and terminology follow~\cite{odonnell-book}. We will view the (real) vector space of functions on the Boolean hypercube $f :\bn\to\R$ as an inner product space with inner product 
\[\abra{f,g} := \Ex_{\bx\sim\bn}\sbra{f(\bx)\cdot g(\bx)}.\]
We define $\|f\|_2 := \sqrt{\abra{f,f}}.$ Note that if $f\isafunc$, then $\|f\|_2 = 1$.

Given a set $S\sse[n]$, we define the \emph{parity function on $S$}, written $\chi_S \isafunc$, as 
$\chi_S(x) := \prod_{i\in S} x_i$
with $\chi_{\emptyset} \equiv 1$ by convention. 
It is easy to check that $(\chi_S)_S$ forms an orthonormal basis with respect to the above inner product. In particular, every function $f\isafunc$ can be written as 
\[f = \sum_{S\sse[n]}\wh{f}(S)\chi_S\]
where $\wh{f}(S) := \abra{f,\chi_S}$.
This decomposition can be viewed as a ``Fourier decomposition'' of $f$. It is not too difficult to see that Parseval's and Plancharel's formulas hold in this setting:
\[\abra{f,f} = \sum_{S\sse[n]} \wh{f}(S)^2 \qquad\text{and}\qquad \abra{f,g} = \sum_{S\sse[n]}\wh{f}(S)\wh{g}(S).\]
It is also immediate that 
$\Ex_{}\sbra{f(\bx)} = \wh{f}(\emptyset)$ and $\Varx_{}[f(\bx)] = \sum_{S\neq\emptyset}\wh{f}(S)^2$ for $\bx\sim\bn$. 

\begin{notation}
	Given $k\in\{0,\ldots,n\}$, and a function $f:\bn\to\R$, we write
	\[\bW^{\leq k}[f] := \sum_{|S|\leq k}\wh{f}(S)^2\]
	with $\bW^{\geq k}[f]$ and $\bW^{=k}[f]$ defined similarly.
\end{notation}

A useful operator in the analysis of Boolean functions is the \emph{Bonami--Beckner noise operator}, which we proceed to define next:

\begin{definition} \label{def:noise}
	Fix $\rho\in[0,1]$. For a given $x\in\zo^n$, we write $\by\sim N_\rho(x)$ to mean a draw of $\by\in\bn$ where each bit $\by_i$ is drawn as follows: 	
	\[N_{\rho}(x) := \begin{cases}
	 x_i & \text{with probability}~\rho\\
	 +1 & \text{with probability}~\frac{1-\rho}{2}\\
	 -1 & \text{with probability}~\frac{1-\rho}{2}
	 \end{cases}.
	\]
	If $\by\sim N_\rho(x)$, we will sometimes say that $\by$ is $\rho$-correlated with $x$. Given a function $f:\bn\to\R$, we define the \emph{noise operator} $\T_\rho$ as
	\[\T_\rho f(x) := \Ex_{\by\sim N_\rho(x)}\sbra{f(\by)}.\]
\end{definition}

It is a standard fact that the noise operator diagonalizes the parity basis:

\begin{fact}[Proposition~2.47 of~\cite{odonnell-book}] \label{fact:noise-fourier-decomp}
	Given a function $f:\bn\to\R$, we have 
	\[\T_\rho f = \sum_{S\sse[n]} \rho^{|S|}\wh{f}(S)\chi_S.\]
\end{fact}

\subsection{Coordinate Oracles}
\label{subsec:oracles}

The tolerant junta testers of De et al.~\cite{DMN19} and Iyer et al.~\cite{ITW21} rely on the notion of \emph{approximate coordinate oracles} for a function $f\isafunc$, assuming it does not depend on too many coordinates. We will require Corollary~4.7 of \cite{ITW21} which in turn builds on Lemma~3.6 of \cite{DMN19}; we reproduce it below for convenience. 

Throughout, we assume that we have query access to an underlying function $f\isafunc$ and have some (fixed) parameter $k\in\N$.

\begin{proposition} \label{prop:coord-oracles}
	Let $\eps,\delta >0$. There exists a non-adaptive algorithm~\coords~that makes $\poly\pbra{k,\frac{1}{\eps}, \log{\frac{1}{\delta}}}\cdot\frac{1}{\eta}$ queries to $f$ and outputs an \emph{$\eta$-oracle} $\calF$ for a set of coordinates $S \sse[n]$ which is a collection of Boolean-valued functions with the following guarantee:
	\begin{enumerate}
		\item With probability at least $1-\delta$, for every $i\in S$ there exists a $g\in\calF$ such that $\dist(g, x_i) \leq \eta$;
		\item We have $\dist(f,\calJ_k) - \dist(f,\calJ_k(S)) \leq \epsilon$ where $\calJ_k(S)$ denotes the family of $k$-juntas whose relevant coordinates are a subset of $S$; and
		\item For any algorithm $A$ that makes $q$ queries to $\calF$, we may assume that we actually have perfect access to each coordinate oracle (i.e.~$\eta = 0$ in the first bullet above), up to an additive loss of $\delta$ in confidence and a multiplicative overhead of $\poly\pbra{\log q + \log\frac{1}{\delta}}$ in query complexity.
	\end{enumerate}
\end{proposition}

We note that we do not have an \emph{explicit} description of the coordinates in $S$; from an information-theoretic standpoint, this would require query-complexity $\Omega(n)$. We instead have \emph{implicit} access to the coordinates in $S$ (cf.~\cite{servedio2010testing} and the references therein).

Thanks to the second bullet in \Cref{prop:coord-oracles}, it suffices for us to only consider juntas on the $\poly(k)$ many coordinates $S$. {Furthermore, as our algorithm will only make $\exp\pbra{\wt{O}_\epsilon(\sqrt{k})}$ many queries, we can assume (thanks to the third bullet in \Cref{prop:coord-oracles}) throughout the rest of the paper that we have perfect access to all the coordinates in $S$.}

\subsection{Flat Polynomials}
\label{subsec:flat-polys}

Underlying our results are constructions of ``flat'' polynomials, originally developed by Linial and Nisan~\cite{LinialNisan:90} and Kahn, Linial, and Samorodnitsky~\cite{kahn1996inclusion} to prove \emph{approximate inclusion-exclusion} bounds. 
We will require the following construction due Kahn et al.~\cite{kahn1996inclusion}:

\begin{lemma}[Theorem~2.1 of \cite{kahn1996inclusion}]
\label{lem:flat-polynomials-fine}
Fix integers $r, N$ with $2 \sqrt{N} \leq r \leq N$. Then there exists a polynomial $p:\R\to\R$ of degree at most $r$ with the following properties:
\[
\text{(i)}~~p(0)=0 
\qquad\text{and}\qquad
\text{(ii)}~~\max_{i\in[n]} |p(i) - 1| \leq 2 \exp\pbra{-\Omega \left( \frac{r^2}{N \log(N)} \right) }.
\]
\end{lemma}

It will also be important for us that these polynomials do not blow up too much for values greater than $N$ and furthermore that they do not have large coefficients. Towards this, we prove the following:

\begin{lemma}
\label{lem:flat-polynomial-oos-bound}
Let $r,N$ be integers with $r \leq N$ and suppose that $p: \mathbb{R} \rightarrow \mathbb{R}$ is a polynomial of degree $r$ such that $|p(i)| \leq 2$ for all $i = 1, \ldots, N$ with $p(0) = 0$. For any $\ell \geq N$, $p(\ell) \leq 4\ell^r$. 
\end{lemma}

\begin{lemma}
\label{lem:flat-polynomial-coefficient-bound}
Let $r,N$ be integers with $r \leq N$ and suppose that $p: \mathbb{R} \rightarrow \mathbb{R}$ is a polynomial of degree $r$ such that $|p(i)| \leq 2$ for all $i = 1, \ldots, N$ with $p(0) = 0$. Moreover, set $\alpha^{r,N}_i$ such that 
\[p(x) = \sum_{i=1}^r \alpha_i^{r,N} \binom{x}{i}\qquad\text{where}\qquad\binom{x}{i} := \frac{x(x-1)...(x-i+1)}{i!}.\] 
Then we have $|\alpha_i^{r,N}| \leq 2r^r$.
\end{lemma}

We prove~\Cref{lem:flat-polynomial-oos-bound,lem:flat-polynomial-coefficient-bound} in~\Cref{appendix:polynomial}.

\subsection{Numerical Differentiation}
We will also need some standard results about numerical differentiation. 
In particular, we will be interested in rapidly converging backwards difference differentiation formulas. While these are undoubtably known, we are unaware of a reference and include a proof in~\Cref{appendix:diff}.

\begin{theorem}
\label{thm:numerical-diff}
Let $\ell$ and $t$ be positive integers and $f$ be a smooth function. There exists coefficients $\beta_0, ..., \beta_{2\ell-1}$ such that for any $\delta$
	\[\left| \sum_i \beta_i f(x - i \delta) - \frac{d^\ell f}{dx^\ell}(x) \delta^\ell \right| \leq (2 \ell)^{3\ell + 1} \delta^{2\ell} \max_{\xi \in [x- 2 \ell \delta, x]} \left| \frac{d^{2\ell}f}{dx^{2\ell}}(\xi_i) \right| \]
Moreover, we have that $|\beta_i| \leq (2 \ell)^{3 \ell}$ for all $i$.
\end{theorem}

\section{Local Estimators}
\label{sec:local-mean-estimation}

Our approach will focus on a statistical estimation problem to tolerantly test juntas. We start with the following definition:

\begin{definition}[Statistic]
	A statistic is a function $S$ that maps a function on the Boolean hypercube $f: \bn\rightarrow \mathbb{R}$ to a real number $\mathbb{R}$.
\end{definition}

We will especially be interested in \emph{local} estimators, i.e. estimators which only receive access to the values of a function $f:\bn\to\R$ restricted to a Hamming ball of small radius.  

\begin{definition}[$r$-local estimator] \label{def:local-est}
Given a function $f: \bits\to\mathbb{R}$, $x \in \bits^n$, and a positive integer $r \leq n$, an \emph{$r$-local estimator $\calE$} takes as input the values of the function $f$ restricted to the Hamming ball $B(x, r)$---which we will denote by $f|_{B(x,r)}$---and outputs a real number. Additionally, for $\tau\geq0$, we will say a $r$-local estimator $\calE: \bits^{B(x,r)} \to \R$ \emph{$\tau$-approximates} a statistic $S$ if
\[\left|\Ex_{\bx\sim\bn}\sbra{\calE(f|_{B(\bx,r)})} - S(f)\right|\leq\tau.\]
Finally, we say that $\calE$ if $\kappa$-bounded if its range is $[-\kappa,\kappa]$.
\end{definition}

\subsection{Locally Estimating $|\E[f]|$}
\label{subsec:high-precision-estimates}

The following lemma establishes the existence of a local estimator for the absolute value of the mean of a Boolean function. (Note that the $0$-local estimator $\calE(f|_{B(x,r)}) := f(x)$ is an excellent estimator for the mean $\E[f]$ but is the identically-$1$ estimator for $|\E[f]|$.)

\begin{lemma} \label{lemma:no-noise-high-precision-estimate}
	Given a Boolean function $f:\bn\to\R$ and a positive-integer $r \geq \Omega(\sqrt{n})$, there exists an $r$-local estimator that $\tau$-approximates $|\Ex[f]|$ and is $O(1)$-bounded, where 
	\[\tau := 2 \exp   \pbra{ \Omega \pbra{ \frac{-r^2}{n \log(n)} } } .\]
\end{lemma} 

The proof of~\Cref{lemma:no-noise-high-precision-estimate} is identical to that of \Cref{lem:smoothed-fine-estimator} and is hence omitted.\footnote{We remark that \Cref{lem:smoothed-fine-estimator} doesn't require the estimator to be bounded, but if desired we can get an $O(1)$-bounded estimator by thresholding the estimator presented in the proof.} 
 \Cref{lemma:no-noise-high-precision-estimate} allows us to tolerantly test $k$ juntas when the total number of variables of the function is $O(k)$.
We prove this as a warmup in~\Cref{sec:warmup} to illustrate our approach to testing juntas. 
The general case of functions $f\isafunc$, however, will require us to smooth the function with noise. 

\begin{lemma}
  \label{lem:smoothed-fine-estimator}
	Let $r$ be a positive integer, and suppose $f:\bn\to[-1,1]$ is a bounded Boolean function such that there exists an integer $t$ for which 
	\[\bW^{=l}[f] \leq \exp\pbra{-\frac{\ell}{t}}.\]
	Then there exists an $r$-local estimator that $O(\tau)$-approximates $|\Ex[f]|$, where 
	\[
	\tau := 2 \exp\pbra{ - \Omega \left( \frac{r}{t \log^2(rt)} \right) }.
	\]
\end{lemma}

Before turning to the proof of \Cref{lem:smoothed-fine-estimator}, we show the following simple lemma.

\begin{lemma}
\label{lem:low-var-expectation}
Suppose that $\bX$ is a random variable with $\Var[\bX] = \sigma^2$, then 
\[\left| \E \left [|\bX| \right] - \left |\E[\bX] \right| \right| \leq \sigma.\]
\end{lemma}

\begin{proof}
Note that by Jensen
	\[\E[|\bX|] \geq |\E[\bX]|. \]
On the other hand,
	\[ \E[|\bX| - |\E[\bX]|] \leq \E[|\bX - \E[\bX]|] \leq \sqrt{\E[(\bX - \E[\bX])^2]} = \sigma \]
which completes the proof.
\end{proof}

\begin{proof}[Proof of \Cref{lem:smoothed-fine-estimator}]
Let $N$ be a parameter we will set later. Let $p^{N}_r(x)$ be the polynomial from \Cref{lem:flat-polynomials-fine} and let $\alpha_i^{r,N}$ denote coefficients such that
\[p_r^N(x) = \sum_{i = 1}^r \alpha_i^{r,N} \binom{x}{i}\]
where $\binom{x}{i} = \frac{x(x-1)(x-2)...(x-r+1)}{r!}$. Take
\[ g(f|_{B(x,r)}) := f(x) - \sum_{i = 1}^r \alpha_i^{r,N} \sum_{S \subseteq [n]: |S| = i} \frac{\partial f}{\partial x_S} (x) \chi_S(x).\]

Note that this can be computed by only querying $f$ on $B(x,r)$ as

\[\frac{\partial f}{\partial x_S}(x)\chi_S(x) = \frac{1}{2^{|S|}} \cdot \left( \sum_{T \subseteq S}   (-1)^{|T|} f(x^{\oplus T})\right).\]

Writing $f = \sumS \wh{f}(S)\chi_S$, it is easy to check that
\begin{align}
	g(f|_{B(x,r)}) &= f(x) - \sum_{i=1}^r \alpha_i^{r,N} \sum_{S\sse[n]: |S| = i} \sum_{T\supseteq S}\wh{f}(T)\chi_T(x) \nonumber\\
	&= f(x) - \sum_{S\neq \emptyset} p_r^N(|S|)\wh{f}(S)\chi_S(x) \nonumber\\
	&= \E[f] + \sum_{S\neq\emptyset} \pbra{1 - p_r^{N}(|S|)} \wh{f}(S)\chi_S(x). \label{eq:united}
\end{align}
where we used the fact that $\wh{f}(\emptyset) = \E[f]$. It is immediate from the above that 
\[\Ex_{\bx\sim\bn} \left [ g \left( f|_{B(\bx,r)} \right) \right] = \E[f]\]
due to orthogonality of $\{\chi_S\}$.
Furthermore, we have 
\begin{align}
\Varx_{\bx\sim\bn}\sbra{g \left( f|_{B(\bx,r)} \right)} 
&= \Ex_{\bx\sim\bn} \left [ \left( g \left( f|_{B(x,r)} \right) - \E[f] \right)^2 \right] \nonumber \\
&= \sum_{\ell = 1}^n (1 - p_r^{N}(\ell))^2 \bW^{=\ell}[f] \label{eq:cowbell} \\
\intertext{where \Cref{eq:cowbell} follows from \Cref{eq:united} via Parseval's formula. We split the sum into two parts:}
\Varx_{\bx\sim\bn}\sbra{g \left( f|_{B(\bx,r)} \right)}
&= \sum_{\ell \leq N} (1 - p_r^{N}(\ell))^2 \bW^{=\ell}[f] + \sum_{\ell > N} (1 - p_r^{N}(\ell))^2 \bW^{=\ell}[f] \nonumber  \\
& \leq 2e^{-\Omega(r^2/N \log(N))}\cdot \Var[f] + 64 {\sum_{\ell > N} \ell^{2r}} e^{-\ell/t} .  \label{eq:potato}
\intertext{where in order to obtain~\Cref{eq:potato}, we used Item~(ii) of~\Cref{lem:flat-polynomials-fine} to bound the first term and bounded the second term using the Fourier tail bounds in the statement of~\Cref{lem:smoothed-fine-estimator} and~\Cref{lem:flat-polynomial-oos-bound}. Recalling that $\Var[f]\leq 1$ and taking 
\[N := 128rt \log(rt),\]
we get}
\Varx_{\bx\sim\bn}\sbra{g \left( f|_{B(\bx,r)} \right)}  & \leq 2e^{-\Omega ( r/(t \log^2(rt)) )} + 64 \sum_{\ell > N} e^{-\ell/2t} \nonumber \\
&\leq 2e^{-\Omega ( r/(t \log^2(rt)) )} + {O\left( t e^{-64r \log(rt)} \right)} \nonumber\\
&\leq e^{-\Omega ( r/(t \log^2(rt)) )} \nonumber
\end{align}
where we also used the fact that the second term above is a geometric series. The result now follows from \Cref{lem:low-var-expectation}.
\end{proof}

\section{Warmup: Tolerantly Testing $f:\bits^{2k} \to \bits$} 
\label{sec:warmup}

As a warmup, we first show how to use local estimators to tolerantly test if a function $f:\bits^{2k}\to\bits$ is $k$-junta. This illustrates the basic idea of our approach while keeping additional technicalities---such as coordinate oracles (cf.~\Cref{subsec:oracles}) and applications of the noise operator (cf.~\Cref{subsec:fourier-101}) to a minimum. 

Recall from \Cref{rem:tol-test-dist-equiv} that tolerant junta testing is equivalent to estimating the distance to being a junta. The key idea of~\Cref{alg:warmup-estimator} is to use our local estimator for $|\E[f]|$ from~\Cref{lemma:no-noise-high-precision-estimate} to estimate the distance to being a junta on each set of $k$ coordinates. (Note that there are at most ${2k \choose k}$ such sets.) Furthermore, because the estimates are local, this allows us to recycle queries efficiently.
\spnote{Do we want explcitly write that $\calE$ is the approximator from the fine estimates section in the algorithm? \shivam{I think so.}}

\begin{algorithm}[t]
\addtolength\linewidth{-2em}

\vspace{0.5em}

\textbf{Input:} Query access to $f: \bits^{2k} \rightarrow \bits$, $\eps \in (0,1]$ \\[0.25em]
\textbf{Output:} An estimate for $\dist(f, \calJ_k)$

\warmbt($f, \epsilon$):
\begin{enumerate}
	\item Set parameters 
	\[m := \Theta \left( \frac{k}{\eps^2} \right) \qquad\text{and}\qquad r := \wt{\Theta}\left( \sqrt{k \log\pbra{\frac{1}{\eps}}}\right).\]
	\item Draw points $\x{1}, \ldots, \x{m} \sim \bits^{2k}$ uniformly and independently at random, and query $f$ on $B(\x{i}, r)$ for each $i\in[m]$.
	\item For each $i\in[m]$ and every subset $J\sse[{2k}]$ with {$|J|= k$}, compute
	\[E_J^{\x{i}} := \calE\pbra{f_{J\to\x{i}}|_{B(\x{i}|_{\overline{J}}, r)}}\]
	where $\calE$ is the estimator from~\Cref{lemma:no-noise-high-precision-estimate}.
	\item Set 
	\[E_J := \frac{1}{m} \sum_{i=1}^m E_J^{\x{i}}\]
	and output $(1-\max_J E_J)/2$.

\end{enumerate}

\caption{Estimating distance to the closest $k$-junta for $f:\bits^{2k}\to\bits$.}
\label{alg:warmup-estimator}
\end{algorithm}

\begin{theorem} \label{thm:warmup}
	Let $T$ be the output of $\warmbt(f, \eps)$. With high probability, we have 
	\[\left|T - \dist(f, \calJ_k)\right| \leq \eps.\]
\end{theorem}

\begin{proof}
We will establish that the following holds for all sets $J\sse[2k]$ where $|J|\leq k$ with high probability:
\[\left| \dist(f,\calJ_J) - \pbra{\frac{1-|E_J|}{2}} \right|\leq \epsilon\]
where $E_J$ is as in~\Cref{alg:warmup-estimator}. Note that theorem immediately follows from this. 

Recall from~\Cref{lemma:no-noise-high-precision-estimate} that our estimator $\calE$ is $O(1)$-bounded, i.e. $|\calE(\cdot)| \leq O(1)$. So by Hoeffding's bound, we have 
\[
\Prx_{\bx^{(1)}, \dots, \bx^{(m)}} \left[ \left| \frac{1}{m} \sum_{i=1}^m \left|E_J^{\bx^{(i)}}\right| - \Ex_{\by \sim \bits^{2k}} \left[\left|E_J^{\by} \right|\right] \right|  \geq \eps/2 \right] \leq \exp \left( - \Omega \left( m \eps^2 \right)\right) \leq 8^{-k}.
\]
Note, however, that
\[
\Ex_{\by \sim \bits^{2k}} \sbra{\abs{E_J^{\by}}} = \Ex_{\by\sim\bits^{2k}}\sbra{ \calE\pbra{  f_{J \rightarrow \by}|_{B(\by|_{\overline{J}}, r)}}}
\]
and that by the guarantee of~\Cref{lemma:no-noise-high-precision-estimate} we have 
\[
\abs{
\Ex_{\by\sim\bits^{2k}}\sbra{ \calE\pbra{  f_{J \rightarrow \by}|_{B(\by|_{\overline{J}}, r)}}} - \Ex_{\by\sim\bits^{2k}}\sbra{ \abs{ \frac{1}{2^k} \sum_{z_{\overline{J}}\in\bits^k} f(z_{\overline{J}}, \by_{J})}}~
}
\leq \frac{\eps}{2}.
\]
The result now follows from a union bound over all sets $J$; recall that there are at most ${2k \choose k} \leq 4^k$
many such sets.
\end{proof}
\section{Tolerantly Testing $f\isafunc$}
\label{sec:ub}

The core idea to tolerantly test juntas in the general case is similar to that of the warmup from~\Cref{sec:warmup}, albeit with coordinate oracles (cf.~\Cref{subsec:oracles}) to reduce the number of relevant coordinates to $\poly(k)$. However, there are some additional road blocks that we will need to overcome; we describe these in~\Cref{subsec:HO-noise,subsec:diff} before presenting the junta tester in~\Cref{subsec:happy}.

\subsection{Hold-Out Noising}
\label{subsec:HO-noise}

The first issue we face in trying to generalize the algorithm from~\Cref{sec:warmup} from $\Theta(k)$ coordinates to $n$ coordinates is due to the number of coordinate oracles we construct. Recall that~\Cref{prop:coord-oracles} outputs $\poly(k)$ coordinate oracles; for concreteness, we will think of this number as $k^{10}$.\footnote{Given adaptivity, the number of oracles can be reduced from $k^{10}$ to $O_\eps(k)$~\cite{ITW21}; however, it is not clear how to implement this procedure non-adaptively.}
Consequently, the naive approach to estimate the mean of the function $f$ will require balls of radius $\sqrt{k^{10}}$. Querying even one of these balls, however, will require $k^{\Omega(k^5)}$ queries---this is much worse than the known $2^{\wt{O}(k)}$-query non-adaptive bound due to~\cite{DMN19}. Thankfully, we can circumvent this by appropriately noising the function, which allows us to appeal to the high-precision estimator (\Cref{lem:smoothed-fine-estimator}) for functions with sufficiently strong Fourier decay. 

In order to non-adaptively noise the function, we consider the following variant of the Bonami--Beckner noise operator from~\Cref{def:noise}, which we call the ``hold-out'' noise operator:

\begin{definition}[Hold-out noise operator]
	Given a set $S \subseteq [n]$ and $\rho \in [0,1]$, we define the \emph{hold-out noise operator} $\T_{\rho}^S$ as 
	\[\T_\rho^S f(x) := \Ex_{\by \sim N_\rho(x)} \left[f(\by) ~|~ \by_S = x_S \right].
	\]
\end{definition}

There will inevitably be a tradeoff between the number of queries we can make and the noise rate; otherwise, we could simply evaluate $\T_0^S f$ to determine the mean and distance to junta.
For short hand, given a set of coordinate oracles 
$\calF = \{g_1, ..., g_{|\calF|}\}$ and a point $x \in \bits^n$, let 
\[\calF(x) \in \bits^{|\calF|}~\text{denote the string}~y\in\bits^{|F|}~\text{with}~y_i=g_i(x).\]
We extend this notation to subsets $S\sse\calF$ of coordinate oracles, writing $S(x)$ for the vector of queries to the coordinate oracles in $S$ on $x$.

\begin{algorithm}[t]
\addtolength\linewidth{-2em}

\vspace{0.5em}

\textbf{Input:} $f \isafunc$, a set of coordinate oracles $\calF$, $x \in \bits^n$, an noise rate $\rho \in [0,1]$, error $\tau \in (0,1]$, and failure probability $\eta \in (0,\frac{1}{2}]$\\[0.25em]
\textbf{Output:} An additive $\tau$-estimate of $\T_\rho^S f(x)$ for all $S \subseteq \calF$ with $|S| = k$ w.p. $1-\eta$. 

\

\textsc{\TpS}($f, x, \rho, \tau, \eta$):
\begin{enumerate}
	\item Set
	\[N := \Theta \left( \frac{|\calF|^3 k^3 \log(1/\eta)}{\tau^2 \rho^{2k}} \right).\]
	\item For $1 \leq i \leq N$, rerandomize $x$ independently with probability $1-\rho$ to generate strings $\by^{(1)}, \ldots , \by^{(N)}$ (cf.~\Cref{def:noise}). Query $f$ on all of these points. 
	\item For each $S \sse \calF$ with $|S| = k$, set 
	\[\wt{\T}_{\rho}^S(x) := \frac{1}{N \rho^k} \sum_{i: S(\by^{(\bi)}) = S(x)} f(\by^{(i)}) \]
\end{enumerate}

\caption{Non-adaptively implementing the hold-out noise operator.}
\label{alg:HO-noise}
\end{algorithm}

\begin{lemma} \label{lemma:HO-noise}
Let $\{\wt{\T}_\rho^S(x)\}_S$ be the outputs of $\TpS(f,x,\rho, \tau, \eta)$ as described in~\Cref{alg:HO-noise}. Then 
	\[\Pr \left[\text{there exists}~S\sse\calF~\text{with}~|S| = k~\text{s.t.}~\left| \wt{\T}_{\rho}^Sf(x) - \T_\rho^Sf(x) \right| > \tau \right] \leq \eta.\]
\end{lemma}

\begin{proof}
	Fix a set $S\sse\calF$ with $|S| = k$. Note that
	\[\Ex_{\by^{(i)}} \left[ f(\by^{(i)})\cdot\mathbf{1}\cbra{S(\by^{(i)}) = S(x)} \right] = \rho^k \T_\rho^S f(x).\]
	By a Hoeffding Bound,
	\begin{align*}
		\Pr_{\by^{(1)}, \ldots, \by^{(N)}} \left[ \left| \sum_{i = 1}^N f(\by_i) \mathbf{1}\cbra{S(\by^{(i)}) = S(x)} - \rho^k N \T_\rho^S f(x) \right| \geq \tau \rho^k N \right] &\leq 2 \exp \left( -\Omega \left( \frac{\rho^{2k} N^2 \tau^2}{N} \right) \right) \\
		&\leq \eta |\calF|^{-k}
	\end{align*}
	The result now follows by taking a union bound over all sets $S \in \binom{\calF}{k}$.
\end{proof}

\subsection{Estimating the Estimator}
\label{subsec:diff}

The hold-out noise operator allows us to smooth the function enough to be able to use balls of radius $\sqrt{k}$ in our local estimator from~\Cref{lem:smoothed-fine-estimator}. However, we immediately run into a second issue: We would like to average out ``irrelevant'' coordinates (namely the coordinates which are not among the coordinate oracles). While this is easy to do adaptively, it is unclear how to implement this non-adaptively. Dealing with all $n$ coordinates is infeasible, since sampling even one Hamming ball of radius $\sqrt{k}$ would require $ n^{\Omega(\sqrt{k})}$ queries. Consequently, we must devise a procedure to \emph{estimate} our estimator from~\Cref{lem:smoothed-fine-estimator} that does not require too many queries.

Recall from~\Cref{lem:smoothed-fine-estimator} that our estimator computes
\[\calE\pbra{f|_{B(x,r)}} = f(x) - \sum_{i = 1}^r \alpha_i^{r,N} \sum_{|S| = i} \frac{\partial f}{\partial x_S} (x) \chi_S(x).\]
As such it suffices to approximate
\[\sum_{|S| = i} \frac{\partial f}{\partial x_S} (x) \chi_S(x). \]
To do this, we compute derivatives of the noise operator $\frac{d^{i}}{d\rho^{i}} \T_\rho f$ at $\rho=1$. Indeed, note that
\begin{align*}
\frac{1}{i!}\cdot\frac{d^{i}}{d\rho^{i}} \T_\rho f(x) \bigg|_{\rho = 1} &= \sum_{|T| \geq i} \frac{|T|(|T|-1)(|T|-2) \dots (|T|- i + 1)}{i!} \wh{f}(T) \chi_T(x) \\ 
	&= \sum_{|T| \geq i} \binom{|T|}{i} \wh{f}(T) \chi_T(x) \\
	&= \sum_{|S| = i} \frac{\partial f(x)}{\partial x_S} \chi_S(x)	
\end{align*}

So it remains to estimate these derivatives by computing them numerically; we will do this via \Cref{thm:numerical-diff}. 

\begin{remark}
	We briefly address why we rely on~\Cref{thm:numerical-diff} rather than using the naive formula 
	\[
		\frac{d^j}{d \rho^j} \T_{\rho} f(x) \bigg |_{\rho = 1} = \lim_{\delta \rightarrow 0} \frac{\sum_{i = 0}^j (-1)^i \binom{j}{i} \T_{1 - i \delta} f(x)}{\delta^j}.
	\]
	In particular, the error in the above expression decays roughly as $O\left(\delta \frac{d^{j+1}}{d\rho^{j+1}}\T_\rho f(x)\big|_{\rho = \xi}\right)$. Consider taking the $\sqrt{k}/2$nd derivative of $x^{\sqrt{k}}$. (We ought to get a good approximation here since the noise we apply does not kill terms at level $\sqrt{k}$.) However, in this case, the error looks like $\exp(\sqrt{k}) \delta$. We can take $\delta = \exp(-\sqrt{k})$, but the denominator in the expression for the derivative would then be $\delta^{\sqrt{k}/2} = \exp(\Omega(k))$. So in order to evaluate this expression we need to estimate $\T_\rho f$ to additive error $\exp(-k)$, which would require too many samples to do naively.
Fortunately, we can rectify this with the speedier convergence guaranteed by~\Cref{thm:numerical-diff}.
\end{remark}

We have the following consequence of~\Cref{thm:numerical-diff}:

\begin{lemma}
\label{lem:numerical-diff-noise-op}
	Fix an integer $\ell$ and a sufficiently large positive integer $k$. Let $f \isafunc$ and let $g = \T_\rho f$ where $\rho \leq 1 - \frac{1}{{k}}$. Then for any positive $\delta \leq O((\ell k)^{-100})$, there exists constants $\gamma_1^{\ell \delta}, \ldots , \gamma^{\ell, \delta}_{2 \ell - 1}$ with $|\gamma_i^{\ell, \delta}| \leq \left( \frac{2\ell}{\delta^{1/3}} \right)^{3\ell}$ such that
	\[\Ex_{\bx \sim \bits^n} \left[ \left( \sum_{i = 0}^{2 \ell - 1} \gamma_i^{\ell, \delta} \T_{1 - i \delta} g(\bx) - \frac{d^\ell}{d\rho^\ell} \T_\rho g(\bx) \right)^2 \bigg |_{\rho = 1} \right] \leq O(\delta^{\ell/6}).\]
\end{lemma}

\begin{proof}
Applying \Cref{thm:numerical-diff} with $f = x^t$, $\gamma_i = \frac{\beta_i}{\delta^\ell}$, and $x = 1$ yields 
	\[\left| \sum_{i = 0}^{2 \ell - 1} \gamma_i (1- i \delta)^t - \frac{t!}{(t-\ell)!} \right| \leq (2 \ell)^{3\ell + 1} \delta^{\ell} t^{2\ell} \]
for any positive integer $t$. Moreover, $|\gamma_i| \leq \left( \frac{2\ell}{\delta^{1/3}} \right)^{3\ell}$. Using these coefficients, we compute 
\begin{align*}
E(x) &:= \sum_{i = 0}^{2 \ell - 1} \gamma_i \T_{1 - i \delta} g(\bx) - \frac{d^\ell}{d\rho^\ell} \T_\rho g(\bx) \bigg|_{\rho = 1} \\
&= \sum_{T \subseteq [n]} \left( \sum_{i = 1}^{2 \ell - 1} \gamma_i (1 - i\delta)^{|T|} - \frac{t!}{(t - \ell)!} \right) \wh{g}(T) \chi_T(x) \\
\end{align*}
By Parseval's formula, we have that
\begin{align*}
\Ex_{\bx} [E(\bx)^2] &= \sum_{j = 0}^n \left( \sum_{i = 1}^{2 \ell - 1} \gamma_i (1 - i\delta)^{j} - \frac{j!}{(j - \ell)!} \right)^2 \bW^{=j}[g] \\
	&\leq \sum_{j = 0}^{\delta^{-1/3}} \left( \sum_{i = 1}^{2 \ell - 1} \gamma_i (1 - i\delta)^{j} - \frac{j!}{(j - \ell)!} \right)^2 \bW^{=j}[g] + \sum_{j = \delta^{-1/3}}^{n} \left( \sum_{i = 1}^{2 \ell - 1} \gamma_i (1 - i\delta)^{j} - \frac{j!}{(j - \ell)!} \right)^2 \bW^{=j}[g] \\
	&\leq \sum_{j = 0}^{\delta^{-1/3}} (2 \ell)^{6\ell + 2} j^{4 \ell} \delta^{2 \ell} \bW^{=j}[g] + \sum_{\delta^{-1/3}}^{n} \left( \sum_{i = 1}^{2 \ell - 1} \gamma_i (1 - i\delta)^{j} - \frac{j!}{(j - \ell)!} \right)^2 \bW^{=j}[g] \\
	&\leq \sum_{j = 0}^{\delta^{-1/3}} (2 \ell)^{6\ell + 2} \delta^{2 \ell/3} \bW^{=j}[g] + \sum_{\delta^{-1/3}}^{n} \left( \sum_{i = 1}^{2 \ell - 1} \gamma_i (1 - i\delta)^{j} - \frac{j!}{(j - \ell)!} \right)^2 \bW^{=j}[g] \\
	&\leq \delta^{\ell/6} + \sum_{j = \delta^{-1/3}}^{n} \left( \sum_{i = 1}^{2 \ell - 1} \gamma_i (1 - i\delta)^{j} - \frac{j!}{(j - \ell)!} \right)^2 \bW^{=j}[g]
\end{align*}
To bound the second term, we note that since $g$ is is the result of applying noise to a Boolean function $f$ the weight at level $j$ is at most $e^{-j/{k}}$ and 
\[\left( \sum_{i = 1}^{2 \ell - 1} \gamma_i (1 - i\delta)^{j} - \frac{j!}{(j - \ell)!} \right)^2  \leq O \left( \frac{(2 \ell)^{10\ell}}{\delta^{2\ell}} j^{2\ell} \right)\leq O \left( \frac{\delta^\ell}{k} j^{70 \ell} \right).\]
when $j \geq \delta^{-1/3}$.  
Combining these, we get that
\[ \left( \sum_{i = 1}^{2 \ell - 1} \gamma_i (1 - i\delta)^{j} - \frac{j!}{(j - \ell)!} \right)^2 \bW^{=j}[g] \leq O \left( \frac{\delta^{\ell}}{k} j^{70 \ell} e^{-j/{k}} \right)\]

Observe
\[\frac{\delta^{\ell}}{k} j^{70 \ell} e^{-j/{k}} \leq \frac{\delta^{\ell}}{k} e^{70 \ell \log(j) -j/{k}} \leq \frac{\delta^{\ell}}{k} e^{-j/(2{k})} \]
assuming that $k$, and hence $j$, is sufficiently large. Thus,
\begin{align*}
\Ex_{\bx} [E(\bx)^2] &\leq \delta^{\ell/6} + \frac{O(\delta^\ell)}{k} \sum_{j = \delta^{-1/3}}^{\infty} e^{-j/(2{k})} \\
&\leq \delta^{{\ell/6}} + \frac{O(\delta^\ell)}{k} \cdot \frac{2 {k}}{1-1/e} \\
&\leq O(\delta^{\ell/6})
\end{align*}
which completes the proof. 
\end{proof}

With this, we now consider the estimator 

\[\calE_{\rho, \delta, r}^\star(f,x) = f(x) - \sum_{\ell = 1}^r \alpha_\ell^{r,N} \left({\frac{1}{\ell!}} \sum_{i = 0}^{2 \ell - 1} \gamma_i^{\ell, \delta} \T_{1 - i \delta^{r/\ell}} \T_{\rho} f \right). \]

It's not hard to see that $|\calE_{\rho, \delta, r}^\star(f,x)|$ will be a good estimator for the absolute value of the mean:

\begin{lemma}
\label{lem:e-star-good}
Let $f \isafunc$, $k$ a sufficiently large positive integer, and $\tau \in (0,\frac{1}{2}]$. Suppose that $r \geq \wt{\Omega}( \sqrt{k \log(1/\tau)})$, $\rho \leq \left(1 - \frac{\sqrt{\log(1/\tau)}}{\sqrt{k}} \right)$, and $\delta \leq O \left( \tau^{12/r} (rk)^{-1000}\right)$, then
	\[\left | \Ex_{\bx} \left[ \left| \calE_{\rho, \delta, r}^\star(f, \bx) \right| \right] - \left| \Ex_{\bx} \left[ f(\bx) \right] \right| \right| \leq \tau\]
\end{lemma}

\begin{proof}
	We first note that by \Cref{lem:flat-polynomial-coefficient-bound}, \Cref{lem:numerical-diff-noise-op},
	and the Cauchy-Schwartz inequality, we get that
	\[\Ex_{\bx} \left[ \left( \calE^\star_{\rho, \delta, r}(f, \bx) - \calE \left( \T_{\rho} f|_{B(\bx, r)} \right) \right)^2 \right] \leq (2 r^2) \cdot (2r)^{2r} \cdot O(\delta^{r/6}) \leq \frac{\tau^2}{1000}. \]
	This in turn implies that
		\[\Ex_{\bx} \left[ \left( \left| \calE^\star_{\rho, \delta, r}(f, \bx) \right| - \left| \calE \left( \T_{\rho} f|_{B(\bx, r)} \right) \right| \right)^2 \right] \leq \frac{\tau^2}{1000}. \]
	Finally, Jensen's inequality yields 
	\[ \left| \Ex_{\bx} \left [ \left|  |\calE^\star_{\rho, \delta, r}(f, \bx) \right|\right] - \Ex_{\bx} \left[ \left |\calE(\T_{\rho} f|_{B(\bx, r)}) \right| \right]\right|  \leq \frac{\tau}{2}. \]
	On the other hand, we have that by \Cref{lem:smoothed-fine-estimator} 
	\[\left | \Ex_{\bx} \left[ \left |\calE(\T_{\rho} f|_{B(\bx, r)}) \right|  \right] - \left| \Ex_{\bx}[f(\bx)] \right| \right| \leq \frac{\tau}{2} \]
	The lemma now follows by the triangle inequality.
\end{proof}

\begin{algorithm}[t]
\addtolength\linewidth{-2em}

\vspace{0.5em}

\textbf{Input:} $f: \bits \rightarrow \{-1,1\}$, a set of coordinate oracles $\calF$, $x \in \bits^n$, an error $\tau \in (0, \frac{1}{2}]$, and a failure probability $\eta \in (0,\frac{1}{2}]$\\[0.25em]
\textbf{Output:} For each $S \in \binom{\calF}{k}$ outputs estimates $\nu_S$ for $|\Ex_{\by \in \bits^{[n] \setminus S}}[f(x|_S \sqcup \by)]|$.

\

\textsc{Estimate-Absolute-Mean}($f, \calF, \tau$):
\begin{enumerate}
	\item Set parameters
	\[r := \wt{\Theta} \left(\sqrt{k \log\left(\frac{1}{\tau}\right)} \right), \quad \rho = \left(1 - \frac{\sqrt{\log(3/\tau)}}{\sqrt{k}} \right), \quad \delta := \Theta (\tau^{12/r} (rk)^{-1000})\] 
	and let $\eta := \tau^{20} 2^{-k \log^3(1/\tau)}$.
	\item For $\ell = 1$ to $r$ and $i = 1$ to $2r-1$: 
	\begin{itemize}
		\item Run $\textsc{Hold-Out-Noise-Evaluations}(f,x,\rho(1 - i \delta^{r/\ell}), \tau (1000 r \delta^{-1})^{-100r}, \eta/(2r^2))$ to generate $\wt{\T}^S_{\rho(1-i\delta^{r/\ell})} f(x)$. 	
	\end{itemize}
	\item Output $\nu_S(x) := \left| f(x) - \sum_{\ell = 1}^r \alpha_\ell^{r,N} \left(\sum_{i = 0}^{2 \ell - 1} \gamma_i^{\ell, \delta} \wt{\T}_{\rho(1 - i \delta^{r/\ell})}^S f(x) \right) \right| $ for all $S \in \binom{\calF}{k}$.
\end{enumerate}

\caption{Estimate Absolute Mean}
\label{alg:eam}
\end{algorithm}

We now conclude with an algorithm (\Cref{alg:eam}) to implement $\calE^\star$ in our junta testing setting.
We also quickly record that the procedure outputs estimates whose means are close to the estimates of $|\calE^\star_{r, \delta, \rho}(f|_{[n]\setminus S}, x)|$.

\begin{lemma}
\label{lem:est-mean-alg-stats}
Suppose $f \isafunc$, $\calF$ a set of coordinate oracles, $k$ is a sufficiently large integer, and $\tau, \eta > 0$ , and let $\nu_S(x)$ be the output of $\textsc{Estimate-Absolute-Mean}(f, \calF, \tau)$. Then have that $|\nu_S(x)| \leq \poly(\tau^{-1}) \cdot \exp( \wt{O}(\sqrt{k \log(1/\tau)}))$ and 
	\[\bigg| \E [ \nu_S ] - |\calE^\star_{r, \delta, \rho}(f|_{[n]\setminus S}, x)| \bigg| \leq \tau.\]
	Moreover, the procedure makes $\poly(\tau^{-1}, k) 2^{\wt{O}(\sqrt{k \log(1/\eps)})}$ queries.
\end{lemma}

\begin{proof}
	We first prove that the estimates $\nu_S$ are bounded. Indeed, note that the estimates from \textsc{Hold-Out-Noise-Evaluations} are at most $(\rho (1 - 2r \delta))^{-k} \leq e^{O(\sqrt{k \log(1/\tau)})}$. It then follows by the bounds on $\alpha_\ell^{r,n}$ and $\gamma_i^{\ell,\delta}$ that $\nu_S$ is at most 
		\[e^{O(\sqrt{k \log(1/\tau)})} \left(\frac{(2 r)^{3r}}{\delta^r} \right) \cdot 2 r^r \leq \tau^{-12} \cdot e^{\wt{O}(\sqrt{k \log(1/\tau)})} \]
	Note that this also bounds $|\calE^\star_{r, \delta, \rho}(f|_{[n]\setminus S}, x)|$.

	We now to turn to prove the bound on the expectation. With probability $1-\eta$, we have that $| \wt{\T}_{\rho(1 - i \delta^{r/\ell})} f(x) - \T_{\rho(1 - i \delta^{r/\ell})} f(x)| \leq \tau (1000 r \delta^{-1})^{-100r}$ for all $i = 1 \dots 2r - 1$ and $\ell = 1 \dots r$. In this case, it follows that 
	\[\left| \nu_S(x) - |\calE^\star_{r, \delta, \rho}(f|_{[n]\setminus S}, x)| \right| \leq \tau \cdot (1000 r \delta^{-1})^{-100r} \left ( \frac{(2 r)^{3r}}{\delta^r} \right) \cdot 2 r^r \cdot (2r^2) \leq \tau/2\]
	
	Using our bounds on $\calE^{\star}_{r,\delta,\rho}$ and $\nu_S(x)$, we then have that
	\[\left| \E \left [\nu_S(x) - |\calE^\star_{r, \delta, \rho}(f|_{[n]\setminus S}, x) \right]| \right| \leq \frac{\tau}{2} + 2 \eta \cdot \tau^{-12} \cdot \exp( \wt{O}(\sqrt{k \log(1/\tau)})) \leq \tau.\]
	Finally, note that for the query bound, we make 
		\[\poly(k, \tau^{-1}) \cdot O \left( \frac{|\calF|^3 k^3 \log(1/\eta)}{\tau^2 (1000 r \delta^{-1})^{-200r} \rho^{2k}}  \right) \leq \poly(k, \tau^{-1}) \cdot 2^{\wt{O}(\sqrt{k \log(1/\tau)})}\]
		which completes the proof. 
\end{proof}

\subsection{Putting It All Together: A $2^{\wt{O}(\sqrt{k \log(1/\eps)})}$ Tester}
\label{subsec:happy}

We finally turn to the proof of~\Cref{thm:fine-tester}.

\begin{algorithm}[t]
\addtolength\linewidth{-2em}

\vspace{0.5em}

\textbf{Input:} $f\isafunc$, an integer $k$, and $\eps \in (0,\frac{1}{2}]$ \\
\textbf{Output:} An estimate for $\dist(f, \calJ_k)$

\

\finebt($f, \eps$):
\begin{enumerate}
	\item Construct coordinate oracles $\calF$ as in \Cref{prop:coord-oracles} with error rate $\eps/3$ and failure probability $\frac{1}{k}$
	\item Draw $m := \frac{ k |\calF| 2^{\wt{O}(\sqrt{k \log(1/\eps)})}}{\eps^{O(1)}}$ points $\bx^{(1)},..., \bx^{(m)} \sim \bn$ uniformly and independently at random
	\item For each point, run $\textsc{estimate-absolute-mean}(f, \calF, \bx^{(i)}, \eps/4)$ to get estimates $\nu_S(x)$.
	\item Set $\bD_S = \frac{1}{m} \sum_i \nu_S(\bx^{(i)})$
	\item Output $\frac{1 - \max_S \bD_S}{2}$
\end{enumerate}

\caption{An algorithm to estimate $k$-junta distance}
\label{alg:fine}
\end{algorithm}

\begin{theorem}
	Let~\finebt{} be as in~\Cref{alg:fine}. Suppose that $f \isafunc$, then
		\[\Pr \left[\left| \finebt{}(f,\eps) - \dist(f, \calJ_k) \right| \geq \eps \right] = o_k(1)\]
	Moreover, \finebt{} makes at most $2^{\wt{O}(\sqrt{k} \log^2(1/\eps))}$ queries.
\end{theorem}

\begin{proof}
By \Cref{prop:coord-oracles}, we have with high probability the distance of $f$ to a junta of $f$ restricted to the coordinates corresponding to $\calF$ differs from $\dist(f, \calJ_k)$ by at most $\eps/4$. Moreover, there are at most $\poly(k,1/\eps)$ such oracles. Now fix a set $S \subseteq \calF$. We'll bound the probability that $\frac{1}{2} (1 - \bD_S)$ differs from $\dist(f,\calJ_S)$ by more than $3\eps/4$. Indeed, we note that by \Cref{lem:e-star-good,lem:est-mean-alg-stats} as well as  the triangle inequality that
	\[\left| \Ex_{\bx^{(i)}} [\nu_S(\bx^{(i)})] - \frac{1}{2^k} \sum_{y \in \bits^S} \left| \Ex_{\bz \sim \bits^{[n] \setminus S}} [f(y \sqcup \bz)] \right | \right| \leq \eps/2. \]
	It then follows that
	\begin{align*}
		\Pr \left[ \frac{1}{2} (1 - \bD_S) - \dist(f, \calJ_S)| \geq 3 \eps/4 \right] &\leq \Pr \left[ | \bD_S - \E[\bD_S] | \geq \eps/4 \right] \\
		&\leq \exp \left( \frac{\eps^2 m }{\poly(\eps^{-1}) 2^{\wt{O}(\sqrt{k \log(1/\eps)})}} \right) \\
		&\leq \frac{1}{k} |\calF|^{-k} 
	\end{align*}
	by a Hoeffding bound. By a union bound, with high probability $\bD_S$ is within $3 \eps/4$ of $\dist(f, \calJ_S)$ for all $S \in \binom{\calF}{k}$, as claimed.
	
	For the query bound, note that we make
	\[\poly(k, \eps^{-1}) \cdot 2^{\wt{O}(\sqrt{k \log(1/\eps)})} \cdot \frac{ k |\calF| 2^{\wt{O}(\sqrt{k \log(1/\eps)})}}{\eps^{O(1)}} \leq \poly(k, \eps^{-1}) 2^{\wt{O}(\sqrt{k \log(1/\eps)})}\]
	queries.	
\end{proof}

\section{Lower Bound}
\label{sec:lb}

We now show that the upper bound obtained in~\Cref{sec:ub} is essentially tight. In particular, we establish the following:

\begin{theorem} \label{thm:lb}
	For every $\eps \in (2^{-O(k)}, k^{-2})$, there exist $\eps_1, \eps_2 \in [0,1/2)$ with $\eps = \eps_2 - \eps_1$ such that every $(\eps_1, \eps_2)$-tolerant $k$-junta tester for functions $f:\bits^{2k}\to\bits$ must make 
	\[\exp({\Omega(\sqrt{k\log(1/k\eps)}}))~\text{queries}\]
	to the function $f$. 
\end{theorem}

Our proof of~\Cref{thm:lb} builds on the constructions of Chen et al.~\cite{chen2023mildly} who proved \Cref{thm:lb} for $\eps_1=0.1$ and $\eps_2 = 0.2$. Note that this theorem implies \Cref{thm:lower-bound} by \Cref{rem:tol-test-dist-equiv} when $\eps < k^{-2}$; on the other hand, if $\eps > k^{-2}$, then \Cref{thm:lower-bound} holds by simply using the $2^{\Omega(\sqrt{k})}$ lower bound from \cite{chen2023mildly}.

\subsection{The $\Dyes$ and $\Dno$ Distributions}
\label{sec:lb-yes-no}

As in~\cite{chen2023mildly}, we start with some objects that we need in the construction of the two distributions $\Dyes$ and $\Dno$. We partition the variables $x_1,\cdots,x_{2k}$ into \emph{control} variables and \emph{action} variables as follows: Let $A\sse[n]$ be a fixed subset of size $k$, and let $C = [2k]\setminus A$. We refer to the variables $x_i$ for $i\in C$ as \emph{control variables} and the variables $x_i$ for $i\in A$ as \emph{action variables}.

\usetikzlibrary{decorations.pathreplacing,angles,quotes}
\usetikzlibrary{shapes.geometric}
\usetikzlibrary{patterns}

\begin{figure}[t]
\centering
\begin{tikzpicture}[scale=0.86]

	
	\draw[-] (0, -1.5) -- (-1.5, 0) -- (0, 1.5) -- (1.5, 0) -- (0, -1.5);

	
	
	\node (heads) at (5.5, 2.75){\includegraphics[width=1.1cm]{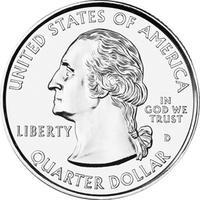}};
	\node (bh) at (5.5, 1.5){If $\bb_1(x_{\bC}) = +1$:};
	
	\fill[color=purple!20!blue!40!white!70] (7.5, 2.5) -- (9, 4) -- (10.5, 2.5);
	\fill[color=purple!20!blue!40!white!70] (10, 2) -- (9, 1) -- (8, 2);
	
	\draw[-] (9, 1) -- (7.5, 2.5) -- (9, 4) -- (10.5, 2.5) -- (9, 1);
	\draw[] (10, 2) -- (8, 2);
	\draw[] (7.5, 2.5) -- (10.5, 2.5);
	
	\node() at (9, 3.25){$+1~$};
	\node() at (9, 2.25){$-1~$};
	\node() at (9, 1.6){$+1~$};
	
	
	\node (tails) at (5.5, -1.75){\includegraphics[width=1.1cm]{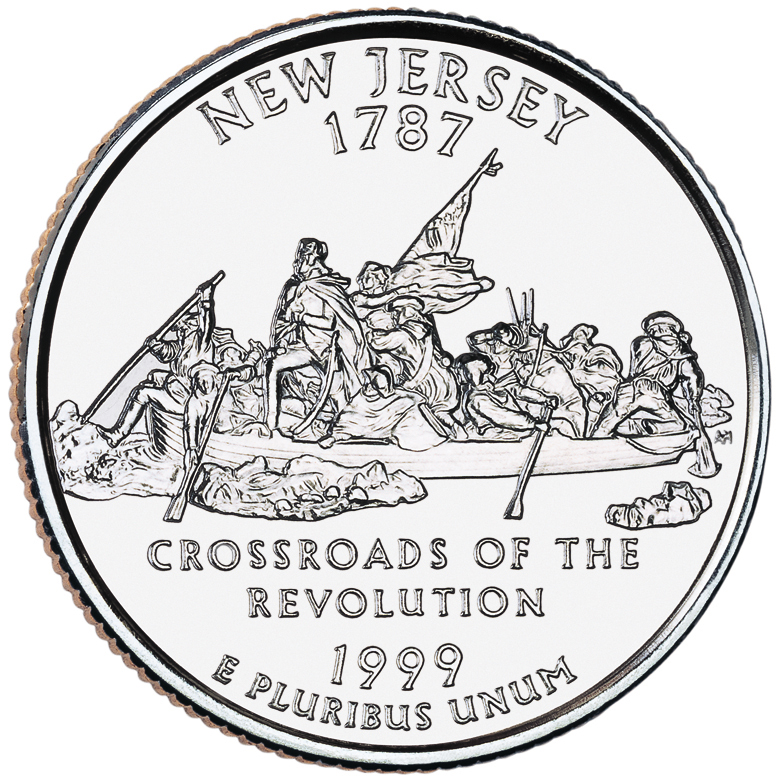}};
	\node (bt) at (5.5, -3){If $\bb_1(x_{\bC}) = -1$:};
	
	\draw[-] (9, -3.5) -- (7.5, -2) -- (9, -0.5) -- (10.5, -2) -- (9, -3.5);	
	\draw[] (10, -2.5) -- (8, -2.5);
	\draw[] (7.5, -2) -- (10.5, -2);
	
	\node() at (9, -1.25){$-1~$};
	\node() at (9, -2.25){$-1~$};
	\node() at (9, -2.9){$-1~$};
	
	\draw[decoration={brace,raise=4pt,amplitude=3pt}, decorate, line width=0.25mm,] (15, 2) -- node[right=-25pt, align=center, text width = 2.5cm] { $\eps$} (15, 1);

	\node (header-1-coin) at (9,6){\includegraphics[width=1.1cm]{images/quarter-front}};
	\node (header-1) at (9, 4.75) {If $\bb_2(x_{\bC}) = +1$:};
	
	\fill[color=purple!20!blue!40!white!70] (15, 2.5) -- (12, 2.5) -- (12.5, 2) -- (14.5, 2) -- (15, 2.5);
	
	\draw[-] (13.5, 1) -- (12, 2.5) -- (13.5, 4) -- (15, 2.5) -- (13.5, 1);
	\draw[] (14.5, 2) -- (12.5, 2);
	\draw[] (15, 2.5) -- (12, 2.5);
	
	\node() at (13.5, 3.25){$-1~$};
	\node() at (13.5, 2.25){$+1~$};
	\node() at (13.5, 1.6){$-1~$};
	
	\node (header-0-coin) at (13.5,6){\includegraphics[width=1.1cm]{images/jersey-quarter-back}};
	\node (header-0) at (13.5, 4.75) {If $\bb_2(x_{\bC}) = -1$:};
	\fill[color=purple!20!blue!40!white!70] (13.5, -3.5) -- (12, -2) -- (13.5, -0.5) -- (15, -2) -- (13.5, -3.5);
	\draw[-] (13.5, -3.5) -- (12, -2) -- (13.5, -0.5) -- (15, -2) -- (13.5, -3.5);
	\draw[] (14.5, -2.5) -- (12.5, -2.5);
	\draw[] (12, -2) -- (15, -2);
	\node() at (13.5, -1.25){$+1~$};
	\node() at (13.5, -2.25){$+1~$};
	\node() at (13.5, -2.9){$+1~$};

	
	\node () at (0,0) [circle,fill,inner sep=1pt]{};
	\node (xc) at (0,-0.25) {$x_{\bC}$};
	
	\node (control) at (0, -4.75){$\bits^{\bC} \equiv \bits^{k}$};
	\node (action) at (11.25, -4.75){$\bits^{\bA} \equiv \bits^{k}$};
	
\end{tikzpicture}

%

\bigskip

\caption{A draw of $\bfyes\sim\Dyes$. The left hand side depicts the control subcube $\bits^{\bC}$, the right hand side depicts an action subcube $\bits^{\bA}$ and the dashed lines indicate Hamming weight levels. The functions $\bb_1,\bb_2:\bits^{\bC}\to \bits$ are random functions on the control subcube.}
\label{fig:yes}
\end{figure}

%
\usetikzlibrary{decorations.pathreplacing,angles,quotes}
\usetikzlibrary{shapes.geometric}
\usetikzlibrary{patterns}

\begin{figure}[t]
\centering
\begin{tikzpicture}[scale=0.86]

	
	\draw[-] (0, -1.5) -- (-1.5, 0) -- (0, 1.5) -- (1.5, 0) -- (0, -1.5);

	
	
	\node (heads) at (5.5, 2.75){\includegraphics[width=1.1cm]{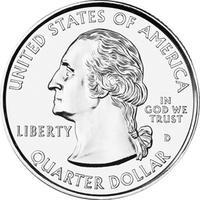}};
	\node (bh) at (5.5, 1.5){If $\bb_1(x_{\bC}) = +1$:};
	
	\fill[color=purple!20!blue!40!white!70] (10, 2) -- (9, 1) -- (8, 2);
	
	\draw[-] (9, 1) -- (7.5, 2.5) -- (9, 4) -- (10.5, 2.5) -- (9, 1);
	\draw[] (10, 2) -- (8, 2);
	\draw[] (7.5, 2.5) -- (10.5, 2.5);
	
	\node() at (9, 3.25){$-1~$};
	\node() at (9, 2.25){$-1~$};
	\node() at (9, 1.6){$+1~$};
	
	
	\node (tails) at (5.5, -1.75){\includegraphics[width=1.1cm]{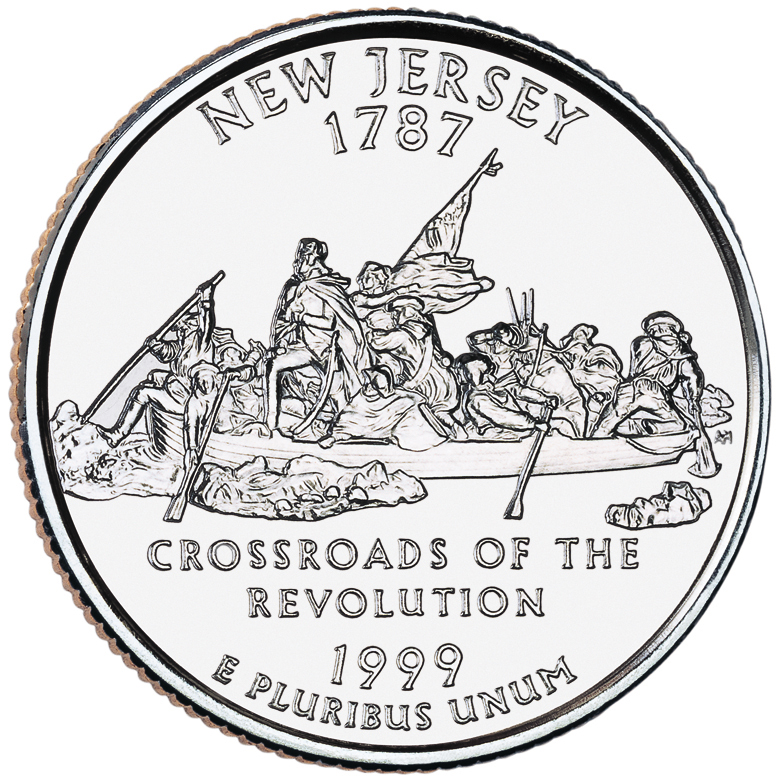}};
	\node (bt) at (5.5, -3){If $\bb_1(x_{\bC}) = -1$:};
	
	\fill[color=purple!20!blue!40!white!70] (9,-0.5) -- (7.5,-2) -- (10.5,-2);
	\draw[-] (9, -3.5) -- (7.5, -2) -- (9, -0.5) -- (10.5, -2) -- (9, -3.5);	
	\draw[] (10, -2.5) -- (8, -2.5);
	\draw[] (7.5, -2) -- (10.5, -2);
	
	\node() at (9, -1.25){$+1~$};
	\node() at (9, -2.25){$-1~$};
	\node() at (9, -2.9){$-1~$};
	

	\node (header-1-coin) at (9,6){\includegraphics[width=1.1cm]{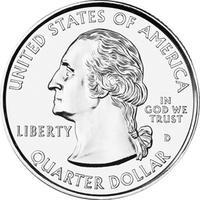}};
	\node (header-1) at (9, 4.75) {If $\bb_2(x_{\bC}) = +1$:};
	
	\fill[color=purple!20!blue!40!white!70] (13.5, 4) -- (12, 2.5) -- (12.5, 2) -- (14.5, 2) -- (15, 2.5);
	
	\draw[decoration={brace,raise=4pt,amplitude=3pt}, decorate, line width=0.25mm,] (15, 4) -- node[right=-15pt, align=center, text width = 2.5cm] {$1-\eps$} (15, 2);
	
	\draw[-] (13.5, 1) -- (12, 2.5) -- (13.5, 4) -- (15, 2.5) -- (13.5, 1);
	\draw[] (14.5, 2) -- (12.5, 2);
	\draw[] (15, 2.5) -- (12, 2.5);
	
	\node() at (13.5, 3.25){$+1~$};
	\node() at (13.5, 2.25){$+1~$};
	\node() at (13.5, 1.6){$-1~$};
	
	\node (header-0-coin) at (13.5,6){\includegraphics[width=1.1cm]{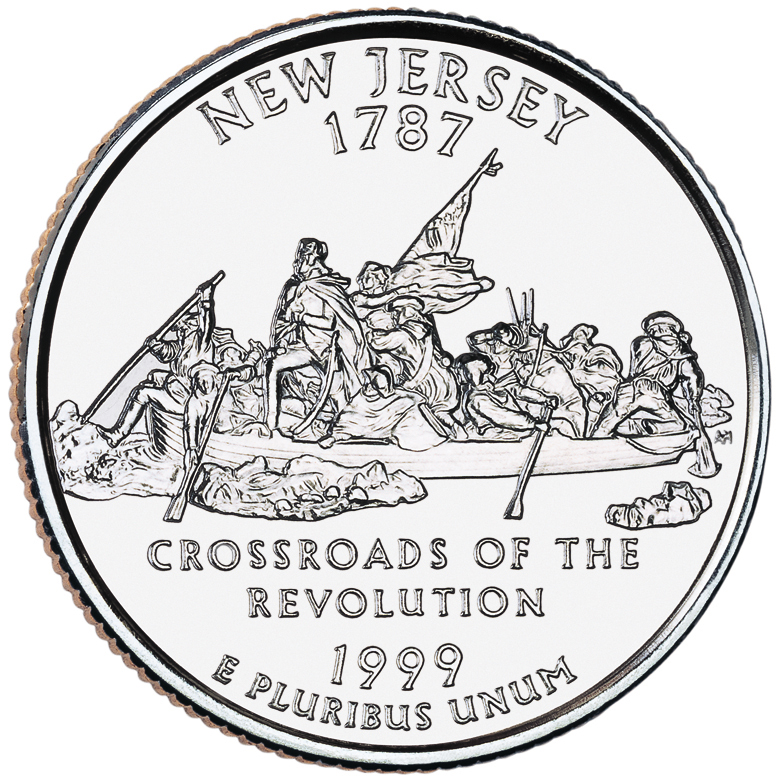}};
	\node (header-0) at (13.5, 4.75) {If $\bb_2(x_{\bC}) = -1$:};
	\fill[color=purple!20!blue!40!white!70] (15, -2) -- (13.5, -3.5) -- (12, -2);
	\draw[-] (13.5, -3.5) -- (12, -2) -- (13.5, -0.5) -- (15, -2) -- (13.5, -3.5);
	\draw[] (14.5, -2.5) -- (12.5, -2.5);
	\draw[] (12, -2) -- (15, -2);
	\node() at (13.5, -1.25){$-1~$};
	\node() at (13.5, -2.25){$+1~$};
	\node() at (13.5, -2.9){$+1~$};

	
	\node () at (0,0) [circle,fill,inner sep=1pt]{};
	\node (xc) at (0,-0.25) {$x_{\bC}$};
	
	\node (control) at (0, -4.75){$\bits^{\bC} \equiv \bits^{k}$};
	\node (action) at (11.25, -4.75){$\bits^{\bA} \equiv \bits^{k}$};
	
\end{tikzpicture}

\bigskip

\caption{A draw of $\bfno\sim\Dno$. Our conventions are as in~\Cref{fig:yes}.}
\label{fig:no}
\end{figure}


\begin{notation} \label{notation:feel-the-bern}
	Given a point $x\in\bits^{2k}$, we write $|x| := |\#\{i : x_i = 1 \}|$ to denote the Hamming weight of the point $x$.
\end{notation}

We first define some auxiliary functions over the corresponding action subcubes $\bits^A$ which we will later use in the definition of $\Dyes$ and $\Dno$. Let $0 \leq \Delta \leq k/2$ be a parameter we will set later; we choose $\Delta$ so as to ensure 
\begin{equation} \label{eq:choice-of-delta}
	\eps := \Prx_{\bx \sim \bits^k}\sbra{|\bx| \leq \frac{k}{2} - \Delta}
\end{equation}
where $\eps$ is as in the statement of~\Cref{thm:lb}.
Let $h^{(+,0)},h^{(+,1)},h^{(-,0)}$ and $h^{(-,1)}$ be Boolean functions over $\bits^A$ defined as follows:
{
\begin{equation*}
\bh^{(+,-)}(x_A) =
	\begin{cases}
       \ -1 & |x_A|\geq\frac{k}{2} \\[0.5ex]
       \ -1 & |x_A|\in (\frac{k}{2}- \Delta, \frac{k}{2}), \\[0.5ex]
       \ -1 & |x_A|\leq \frac{k}{2}-\Delta.
    \end{cases}
\end{equation*}
\begin{equation*}
\bh^{(+,+)}(x_A) =
	\begin{cases}
       \ +1 & |x_A|\geq\frac{k}{2} \\[0.5ex]
       \ -1 & |x_A|\in (\frac{k}{2}- \Delta, \frac{k}{2}), \\[0.5ex]
       \ +1 & |x_A|\leq \frac{k}{2}-\Delta.
    \end{cases}
\end{equation*}
\begin{equation*}
\bh^{(-,-)}(x_A) =
	\begin{cases}
       \ -1 & |x_A|\geq\frac{k}{2} \\[0.5ex]
       \ -1 & |x_A|\in (\frac{k}{2}- \Delta, \frac{k}{2}), \\[0.5ex]
       \ +1 & |x_A|\leq \frac{k}{2}-\Delta.
    \end{cases}
\end{equation*}
\begin{equation*}
\bh^{(-,+)}(x_A) =
	\begin{cases}
       \ +1 & |x_A|\geq\frac{k}{2} \\[0.5ex]
       \ -1 & |x_A|\in (\frac{k}{2}- \Delta, \frac{k}{2}), \\[0.5ex]
       \ -1 & |x_A|\leq \frac{k}{2}-\Delta.
    \end{cases}
\end{equation*}
}
We now turn to the definitions of the ``yes'' and ``no'' distributions. 
To draw a function $\bm{f}_{\yes}\sim \Dyes$, we first sample a set $\bA\subset [2k]$ of size $k$ uniformly 
  at random, set $\bC=[2k]\setminus \bA$, and sample two Boolean functions $\bb_1, \bb_2$ over $\bits^{\bC}$ uniformly at random. 
Then the Boolean function $\bm{f}_{\yes}$ over $\bits^{2k}$  is defined using $\bA$ and $\bb_1, \bb_2$ as follows:
\[
\bm{f}_{\yes}(x)=\begin{cases}
\bh^{(+,-)}(x_\bA)  & \bb_1(x_{\bC})=-1~\text{and}~\bb_2(x_{\bC})=+1\\[0.3ex]
- \bh^{(+,-)}(x_\bA)  & \bb_1(x_{\bC})=-1~\text{and}~\bb_2(x_{\bC})=-1\\[0.3ex]
\bh^{(+,+)}(x_\bA)  & \bb_1(x_{\bC})=+1~\text{and}~\bb_2(x_{\bC})=+1\\[0.3ex]
- \bh^{(+,+)}(x_\bA)  & \bb_1(x_{\bC})=+1~\text{and}~\bb_2(x_{\bC})=-1
\end{cases}.
\]

To draw $\bm{f}_{\no}\sim \Dno$, we first sample  $\bA$, $\bb_1$, and $\bb_2$ in the same way as in $\Dyes$,
  and  $\bm{f}_{\no}$ is defined as
\[
\bm{f}_{\no}(x)=\begin{cases}
\bh^{(-,+)}(x_\bA)  & \bb_1(x_{\bC})=-1~\text{and}~\bb_2(x_{\bC})=+1\\[0.3ex]
- \bh^{(-,+)}(x_\bA)  & \bb_1(x_{\bC})=-1~\text{and}~\bb_2(x_{\bC})=-1\\[0.3ex]
\bh^{(-,-)}(x_\bA)  & \bb_1(x_{\bC})=+1~\text{and}~\bb_2(x_{\bC})=+1\\[0.3ex]
- \bh^{(-,-)}(x_\bA)  & \bb_1(x_{\bC})=+1~\text{and}~\bb_2(x_{\bC})=-1\\[0.3ex]
\end{cases}.
\] 

See Section~6.1 of the full version of the paper for figures illustrating $\Dyes$ and $\Dno$. Recall also that at this point we have not yet specified the parameters $\Delta$ (or equivalently---recalling~\Cref{eq:choice-of-delta}---the parameter $\eps$). 

\begin{proposition} \label{prop:yes-dist}
	With probability $1-o_k(1)$, a function $\bfyes\sim\Dyes$ is 
	\[\pbra{\frac{1}{4} - \frac{\eps}{2} \pm \exp\pbra{-O(k)}} \text{-close to a}~k\text{-junta}.\]
\end{proposition}

\begin{proof}
	Consider the $(n/2$)-junta $g\isazofunc$ given by 
	\[g(x) = 
	\begin{cases}
		-1  & \bb_1(x_{\bC})=-1~\text{and}~\bb_2(x_{\bC})=+1\\[0.3ex]
		+1  & \bb_1(x_{\bC})=-1~\text{and}~\bb_2(x_{\bC})=-1\\[0.3ex]
		+1  & \bb_1(x_{\bC})=+1~\text{and}~\bb_2(x_{\bC})=+1\\[0.3ex]
		-1  & \bb_1(x_{\bC})=+1~\text{and}~\bb_2(x_{\bC})=-1
	\end{cases}.
	\]
	Note that $g$ is a junta on the control subcube. A Chernoff bound on $\bb_1$ together with the definition of $\Dyes$ immediately gives the desired result. 
\end{proof}

\begin{proposition} \label{prop:no-dist}
	With probability $1-o_k(1)$, a function $\bfno\sim\Dno$ is 
	\[\pbra{\frac{1}{4} - e^{-O(k)}}\text{-far from every}~k\text{-junta}.\]
\end{proposition}

\begin{proof}
As in the proof of Lemma~17 of~\cite{chen2023mildly}, we have by a union bound and a Chernoff bound over $\bA$, $\bb_1$, {and $\bb_2$} that with probability at least $1-o_k(1)$, the following holds:
	\begin{quotation}
		\noindent For every $i\in\bC$, there are at least $(0.25 - e^{-O(k)})$-fraction of strings $x \in \bits^{2k}$ such that $x_i = +1$ and $\bfno(x) \not = \bfno(x^{\oplus i})$.
	\end{quotation}
	
	We assume for the rest of the proof that $\fno$ satisfies the above condition. 
	Let $g:\bits^{2k}\to\bits$ be any $k$-junta, and let $I\sse[2k]$ denote its set of influential variables. We now split into two cases:
	\begin{itemize}
		\item If $I \neq C$, then there exists $i\in C$ such that $i\notin I$. Since $g(x) = g(x^{\oplus i})$, however, it follows from the above condition that its distance to being a junta is $\frac{1}{4} - \exp(-O(k))$ since at least one endpoint of each of the influential edges needs to be changed to obtain a junta. 
		\item On the other hand, if $I = C$, then from the construction of $\fno$ it follows (using a Chernoff bound on $\bb_1$) that the distance to being a junta is at least 
		\[\frac{1}{4} + \frac{\eps}{2} - e^{O(k)}.\]
	\end{itemize}
	This completes the proof.
\end{proof}

\subsection{Proof of~\Cref{thm:lb}}
\label{sec:lb-indistinguishability}

We now turn to the proof of~\Cref{thm:lb}. Recall that $\Dyes$ and $\Dno$ are parametrized by $\eps$, or equivalently $\Delta$ (see~\Cref{eq:choice-of-delta}).
We now obtain bounds on $\Delta$ as a function of $\eps$. We start by recalling the following standard tail bound for the Binomial distribution~\cite{ash2012information}:
 
\begin{proposition} \label{prop:bernoulli-tail}
	For $\bX\sim\mathrm{Binomial}(n, p)$ and $k\leq n$, we have 
	\[\Prx\sbra{\bX \leq k} \geq \frac{1}{\sqrt{2n}}\exp\pbra{-n\cdot \mathrm{KL}\pbra{\frac{k}{n}~\bigg\|~p }}\]
	where $\mathrm{KL}(q_1~\|~q_2)$ denotes the KL-divergence between Bernoulli random variables with parameters $q_1$ and $q_2$. In particular, 
	\[
	\mathrm{KL}\pbra{\frac{k}{n}~\bigg\|~p } 
	= \frac{k}{n}\log\frac{k}{np} + \pbra{1-\frac{k}{n}}\log\pbra{\frac{1-\frac{k}{n}}{1-p}}.
	\]
\end{proposition}

Using~\Cref{prop:bernoulli-tail} and recalling~\Cref{eq:choice-of-delta}, we get that 
\begin{equation} \label{eq:berkeley-bowl}
	\eps \geq \frac{1}{2k}\exp\pbra{-k\cdot\mathrm{KL}\pbra{\frac{1}{2} - \frac{\Delta}{k}~\bigg\|~\frac{1}{2}}}. \nonumber
\end{equation}
A routine calculation using the inequality $1+x\leq\exp(x)$ gives that 
\[
\mathrm{KL}\pbra{\frac{1}{2} - \frac{\Delta}{a}~\bigg\|~\frac{1}{2}} \leq \Theta\pbra{\frac{\Delta^2}{a^2}}
\qquad\text{and so}\qquad 
\eps \geq \frac{1}{2k}\exp\pbra{-\Theta\pbra{\frac{\Delta^2}{k}}}
\]
which rearranges to 
\begin{equation} \label{eq:trader-joes}
\Delta \geq \Theta\pbra{\sqrt{k\log\pbra{\frac{1}{k\eps}}}}.
\end{equation}

The remainder of the argument is identical to the lower bound obtained by~\cite{chen2023mildly}; 
we include a brisk proof below for the reader's convenience. 
Our notation and terminology follow~\cite{chen2023mildly} and we refer the reader to Sections 2.1 and 5 of~\cite{chen2023mildly} for background on proving lower bounds for testing algorithms (such as Yao's lemma~\cite{Yao:87,Goldreich17book}). The key observation here---as in~\cite{chen2023mildly}---is the fact that coupling the random variables $\bC, \bA, \bb_1, \bb_2$, and $\br$, a non-adaptive algorithm $\calA$ can distinguish between $\bfyes$ and $\bfno$ only if the following event occurs: 
\begin{quotation}
\noindent There are two points $x, y\in\bits^n$ queried by $\calA$ such that $x_{\bC} = y_{\bC}$, and $|x_{\bA}| \geq \frac{k}{2}$ and $|y_{\bA}| \leq \frac{k}{2} - \Delta$.
\end{quotation}
Call this event $\Bad$. The next lemma shows that $\Bad$ occurs with probability $o_k(1)$

\begin{proposition} \label{prop:thanks-xi}
	Let $\calA$ be a non-adaptive algorithm that makes $\exp({0.000001\sqrt{k\log(1/k\eps)}})$ queries. The probability of the event $\Bad$ is $o_k(1)$.
\end{proposition}

\begin{proof}
	Let $x$ and $y$ be two points queried by $\calA$. For $|x_{\bA}| \geq \frac{k}{2}$ and $|y_{\bA}| \leq \frac{k}{2} - \Delta$ to hold, it must be the case that $x$ and $y$ have Hamming distance at least $\Delta$. However, for any $x$ and $y$ with Hamming distance at least $\Delta$, the probability that $x_{\bC} = y_{\bC}$ is at most $\exp({-0.00001\sqrt{k\log(1/k\eps)}})$. \Cref{prop:thanks-xi} now follows by a union bound over all pairs of points queried by $\calA$.
\end{proof}

\Cref{prop:thanks-xi} now implies \Cref{thm:lb} in the same manner that Lemma 18 of \cite{chen2023mildly} implies Theorem~2 of \cite{chen2023mildly}.

\section{Acknowledgements}
Shivam Nadimpalli is supported by NSF grants IIS-1838154, CCF-2106429, CCF-2211238, CCF-1763970, and CCF-2107187. Shyamal Patel is supported in part by NSF grants IIS-1838154, CCF-2106429, CCF-2107187, CCF-2218677, ONR grant ONR-13533312, and an NSF Graduate Student Fellowship. Part of this work was completed while the authors were visiting the Simons Institute for the Theory of Computing during Summer 2023.  The authors would like Xi Chen, Rocco Servedio, and Anindya De for helpful discussions. They would also like to thank the anonymous STOC reviewers for their feedback.

\begin{flushleft}
\bibliographystyle{alpha}
\bibliography{allrefs}
\end{flushleft}

\appendix
\section{Missing Proofs} 
\label{sec:appendix}

\subsection{Proofs from \Cref{subsec:flat-polys}}
\label{appendix:polynomial}

\begin{proof}[Proof of \Cref{lem:flat-polynomial-oos-bound}]
We prove the statement by Lagrange interpolation. Indeed, note that 
	\[p(x) = \sum_{j = 0}^r p(j) \prod_{i \in \{0, ..., r\} \setminus \{j\}} \frac{x - i}{j - i} \]
and so we have 
	\begin{align*}
		|p(\ell)| &\leq \sum_{j = 0}^r |p(j)| \cdot \left| \prod_{i \in \{0, ..., r\} \setminus \{j\}} \frac{\ell - i}{j - i} \right| \\
		&\leq \sum_{j = 1}^r 2 \left| \prod_{i \in \{0, ..., r\} \setminus \{j\}} \frac{\ell - i}{j - i} \right| \\
		&\leq 4 \ell^r
	\end{align*}
	which completes the proof. 
\end{proof}

\begin{proof}[Proof of \Cref{lem:flat-polynomial-coefficient-bound}]
We prove the statement by induction. We will show that $|\alpha_j^{r,n}| \leq 2(j)^j $. For the base case, note that
	\[ |\alpha_1^{r,n}| = |p(1)| \leq 2\]
by property $(ii)$. Now suppose that the statement holds for all $i < j$. We then note that
	\[2 \geq |p(j)| = \left| \sum_{i=1}^j \alpha_i^{r,n} \binom{j}{i} \right| \]
which implies that
	\[|\alpha_j^{r,n}| \leq 2 + 2\sum_{i=1}^{j-1} i^i \binom{j}{i} \leq 2\sum_{i=0}^{j-1} (j - 1)^i \binom{j}{i} \leq 2j^j\]
as desired.

\end{proof}

\subsection{Proof of \Cref{thm:numerical-diff}}
\label{appendix:diff}

We prove the existence of such coefficients using standard techniques (cf. Section 10.10.2 of~\cite{quarteroni2006numerical}). Before proving the statement, we first recall a few facts from linear algebra.

\begin{definition}
	Given a set of points $x_1, ..., x_n$, the corresponding Vandermonde matrix is defined as
		\[W_n := \begin{pmatrix} 1 & x_1 & \dots & x_1^{n-1} \\ 1 & x_2 & \dots & x_2^{n-1} \\ \vdots & \vdots & \ddots & \vdots \\ 
		 1 & x_n & \dots & x_n^{n-1} \end{pmatrix}.\]
\end{definition}

The following is lemma is easy to check; alternatively, see Exercise~40 from Section 1.2.3 of~\cite{knuth1997art}.

\begin{lemma} \label{lem:vandermonde}
Given a set of points $x_1, ..., x_n$, the corresponding Vandermonde matrix $W_n$ is invertible with inverse $W_n^{-1}$ given by
	\[(W_n^{-1})_{i,j} = \frac{(-1)^{n-i} e_{n-i}(\{x_1, \dots, x_n\} \setminus \{x_j\})}{\prod_{m = 1, m \not = j}^n (x_j - x_m)} \]
where $e_m$ is the $m$th elementary symmetric polynomial, i.e. 
\[e_m(\{y_1, \dots, y_k\}) = \sum_{1 \leq j_1 < ... < j_m \leq k} y_{j_1} y_{j_2} \dots y_{j_m}.\]
\end{lemma}

We now turn to the proof of~\Cref{thm:numerical-diff}:

\begin{proof}[Proof of \Cref{thm:numerical-diff}]
Note that by Taylor's theorem we have that
	\[ f(x - i \delta) = f(x) + \frac{df}{dx}(x) \frac{(-i \delta)}{1!} + \frac{d^2f}{dx^2}(x) \frac{(-i \delta)^2}{2!} + \dots + \frac{d^{2\ell-1}f}{dx^{2\ell-1}}(x) \frac{(-i \delta)^{2\ell-1}}{(2\ell-1)!} + \frac{d^{2\ell}f}{dx^{2\ell}}(\xi_i) \frac{(i \delta)^{2\ell}}{(2\ell)!} \]
for some $\xi_i \in [x-2\ell\delta, x]$. Thus
\begin{multline*}
\sum_i \beta_i f(x - i \delta) = \left (\sum_i \beta_i \right) f(x) + \left (\sum_i \beta_i (-i)  \right) \frac{df}{dx}(x) \frac{\delta}{1!} + \left (\sum_i \beta_i (-i)^2 \right) \frac{d^2f}{dx^2}(x) \frac{\delta^2}{2!} + \dots \\
+ \left (\sum_i \beta_i (-i)^{2\ell-1}  \right)  \frac{ds^{2\ell-1}f}{dx^{2\ell-1}}(x) \frac{\delta^{2\ell-1}}{(2\ell-1)!} + \sum_i \beta_i i^{2\ell} \frac{d^{2\ell}f}{dx^{2\ell}}(\xi_i) \frac{\delta^{2\ell}}{(2\ell)!}.
\end{multline*}

We now set $\beta_i$ such that 
	\[\sum_i \beta_i (-i)^j = \begin{cases} \ell! & j = \ell \\ 0 & \text{otherwise} \end{cases} \]
If we let $W_{2\ell}$ be the Vandermonde matrix for corresponding to the $2l$ points $\pbra{0,-1,...,-(2\ell-1)}$, we see this is equivalent to taking $\beta$ such that 
\[W_{2\ell}^T \beta = \ell!\cdot e_\ell.\]
In particular, such $\beta_i$ exist. Moreover, we have that $\beta = \ell! \left( (W_{2\ell})^{-1}\right)^T e_\ell$ From which, we conclude that
	\[|\beta_i| \leq \ell!\cdot  \| W_{2\ell}^{-1} \|_\infty.\]
Using the closed form expression for $W_{2\ell}^{-1}$ given by~\Cref{lem:vandermonde}, we can bound its largest entry by
\[\max_{i,j} \left| \frac{(-1)^{n-i} e_{2\ell-i}(\{x_1, \dots, x_{2\ell}\} \setminus \{x_j\})}{\prod_{m = 1, m \not = j}^{2\ell} (x_j - x_m)} \right| \leq 2^\ell (2 \ell)!\]
where $x_i = 1 - i$. Thus, we conclude that $|\beta_i| \leq (2 \ell)! (2 \ell)^{\ell}$ as claimed.

Finally, note that
\begin{align*}
\left| \sum_{i=0}^{2\ell-1} \beta_i f(x - i \delta) - \frac{d^\ell f}{dx^{\ell}}(x) \right| &= \left| \sum_{i = 0}^{2\ell-1} \beta_i i^{2\ell} \frac{d^{2\ell}f}{dx^{2\ell}}(\xi_i) \frac{\delta^{2\ell}}{(2\ell)!} \right| \\
&\leq \sum_{i = 0}^{2\ell-1} \left |\beta_i i^{2\ell} \frac{d^{2\ell}f}{dx^{2\ell}}(\xi_i) \frac{\delta^{2\ell}}{(2\ell)!} \right| \\
&\leq (2 \ell)^{3\ell} \delta^{2\ell} \sum_{i = 0}^{2 \ell - 1} \left| \frac{d^{2\ell}f}{dx^{2\ell}}(\xi_i) \right| \\
&\leq (2 \ell)^{3\ell + 1} \delta^{2\ell} \max_{\xi \in [x- 2 \ell \delta, x]} \left| \frac{d^{2\ell}f}{dx^{2\ell}}(\xi_i) \right|
\end{align*}
which completes the proof.
\end{proof}

\end{document}